\documentclass[12pt,twoside]{article}
\usepackage[normalem]{ulem}
\usepackage{srcltx}
\usepackage[mathscr]{eucal}
\usepackage{amsmath,amsfonts,amssymb,amsthm}
\usepackage[matrix,arrow]{xy}
\usepackage{times}
\usepackage{pdfsync}
\usepackage{graphicx}
\usepackage{epstopdf}
\usepackage{multibox}
\usepackage{dcolumn}
\usepackage{hyperref}

\advance\textheight\topskip
\numberwithin{equation}{section} \makeatletter
\@addtoreset{equation}{section}

\newtheorem{prop}{Proposition}[section]

\renewcommand{\tilde}{\widetilde}
\renewcommand{\hat}{\widehat}

\renewcommand{\simeq}{\cong}
\newcommand{\bref}[1]{\textbf{\ref{#1}}}

\newcommand{\p}[1]{|#1|}
\newcommand{\gh}[1]{\mathrm{gh}(#1)}

\newcommand{\dd}{\partial}
\renewcommand{\d}{\partial}

\newcommand{\tensor}{\otimes}

\renewcommand{\geq}{\,{\geqslant}\,}
\renewcommand{\leq}{\,{\leqslant}\,}

\newcommand{\inner}[2]{\langle #1{,}\,#2\rangle}
\newcommand{\binner}[2]{%
  {\langle}\kern-4.15pt{\langle}#1{,}\,#2{\rangle}\kern-4.15pt{\rangle}}
\newcommand{\commut}[2]{[#1{,}\,#2]}

\newcommand{\pb}[2]{\left\{{}#1{},{}#2{}\right\}}

\newcommand{\half}{\mathchoice{%
    \ffrac{1}{2}}{\frac{1}{2}}{\frac{1}{2}}{\frac{1}{2}}}

\newcommand{\ffrac}[2]{\raisebox{.5pt}%
  {\footnotesize$\displaystyle\frac{#1}{#2}$}\kern1pt}

\newcommand{\brst}{\mathsf{\Omega}}

\newcommand{\red}{\mathrm{red}}

\newcommand{\dl}[1]{\mathchoice{\ffrac{\dd}{\dd #1}}{\frac{\dd}{\dd
      #1}}{\ffrac{\dd}{\dd #1}}{\ffrac{\dd}{\dd #1}}}
\newcommand{\dr}[1]{\ffrac{{\overset{\leftarrow}{\partial}}}{ \partial #1}}
\newcommand{\dlf}[1]{\mathchoice{\ffrac{\dd^F}{\dd #1}}{\frac{\dd^F}{\dd
      #1}}{\ffrac{\dd^F}{\dd #1}}{\ffrac{\dd^F}{\dd #1}}}

\newcommand{\dover}[2]{\ffrac{\dd #1}{\dd #2}}

\newcommand{\ddl}[2]{\ffrac{\dd #1}{\dd #2}}

\newcommand{\vdl}[1]{\ffrac{{\delta}}{\delta #1}}

\newcommand{\vddl}[2]{{\ffrac{\delta #1}{\delta #2}}}

\def\cA{\mathcal{A}}

\def\cF{\mathcal{F}}

\def\cH{\mathcal{H}}

\def\cL{\mathcal{L}}

\def\cN{\mathcal{N}}
\def\cO{\mathcal{O}}
\def\cP{\mathcal{P}}

\def\cV{\mathcal{V}}

\def\BGST{Barnich:2004cr}
\def\BGadS{Barnich:2006pc}

\def\Goff{Grigoriev:2006tt}

\def\AGT{Alkalaev:2008gi}
\def\AG{Alkalaev:2009vm}
\def\GD{Grigoriev:1999qz}

\newcommand{\pr}[1]{\left( #1 \right)^F}

\begin{document}

\def\mytitle{First order parent formulation for generic gauge field
  theories}

\pagestyle{myheadings}
\markboth{\textsc{\small Barnich, Grigoriev}}{%
  \textsc{\small Nonlinear parent theory}}
\addtolength{\headsep}{4pt}

\begin{flushright}\small
ULB-TH/10-09	\\
FIAN/TD-06/10
\end{flushright}

\begin{centering}

  \vspace{1cm}

  \textbf{\Large{\mytitle}}

  \vspace{1.5cm}

  {\large Glenn Barnich$^{a,*}$ and Maxim Grigoriev$^b$} 

\vspace{.5cm}

\begin{minipage}{.9\textwidth}\small \it \begin{center}
   $^{a}$Physique Th\'eorique et Math\'ematique\\ Universit\'e Libre de
   Bruxelles\\ and \\ International Solvay Institutes \\ Campus
   Plaine C.P. 231, B-1050 Bruxelles, Belgium \end{center}
\end{minipage}

\vspace{.5cm}

\begin{minipage}{.9\textwidth}\small \it \begin{center}
   $^{b}$Tamm Theory Department, Lebedev Physics
   Institute,\\ Leninsky prospect 53, 119991 Moscow, Russia\end{center}
\end{minipage}

\end{centering}

\vspace{1cm}

\begin{center}
  \begin{minipage}{.9\textwidth}
    \textsc{Abstract}. We show how a generic gauge field theory described by a
BRST differential can systematically be reformulated as a first order parent
system whose spacetime part is determined by the de Rham differential. In the
spirit of Vasiliev's unfolded approach, this is done by extending the original
space of fields so as to include their derivatives as new independent fields
together with associated form fields. Through the inclusion of the antifield
dependent part of the BRST differential, the parent formulation can be used both
for on and off-shell formulations. For diffeomorphism invariant models, the
parent formulation can be reformulated as an AKSZ-type sigma model. Several
examples, such as the relativistic particle, parametrized theories, Yang-Mills
theory, general relativity and the two dimensional sigma model are worked out in
details.
  \end{minipage}
\end{center}


\vfill

\noindent
\mbox{}
\raisebox{-3\baselineskip}{%
  \parbox{\textwidth}{\mbox{}\hrulefill\\[-4pt]}}
{\scriptsize$^*$Research Director of the Fund for
  Scientific Research-FNRS (Belgium).}

\thispagestyle{empty}
\newpage

\begin{small}
{\addtolength{\parskip}{-1.5pt}
 \tableofcontents}
\end{small}
\newpage

\section{Introduction}
\label{sec:introduction}

When dealing with the types of gauge field theories that are of
interest in theoretical high energy physics, it is often useful to
produce equivalent formulations that are local, make rigid symmetries
manifest or allow for an action principle (see
e.g.~\cite{Fierz:1939ix,Singh:1974qz,Fronsdal:1978rb} in the context
of higher spin fields). For instance, there has been a lot of focus on
a first order ``unfolded''
form~\cite{Vasiliev:1988xc,Vasiliev:1988sa,Vasiliev:1990en,%
  Shaynkman:2000ts,Vasiliev:2003ev,Vasiliev:2005zu} of the equations
of motion in the context of higher spin interactions. In this approach
the equations of motion are represented as a free differential algebra
(FDA). The latter structure was originally introduced in
mathematics~\cite{Sullivan:1977fk} and independently in the context of
supergravity~\cite{DAuria:1982nx,Fre:2008qw}.

The characteristic feature of the formulation discussed in the present
paper is that spacetime derivatives enter exclusively through the de
Rham differential acting on form fields. Our aim is the systematic
construction of such a first order formulation for generic gauge
theories. 

In the  linear case~\cite{\BGST,\BGadS}, this problem  has been solved
by using a BRST first quantized approach in combination with a version
of   Fedosov   quantization~\cite{Fedosov-book}.  Various   equivalent
formulations,  including  the  unfolded   one,  are  then  reached  by
reductions  that  correspond  to  the elimination  of  cohomologically
trivial  pairs   on  the   first  quantized  level.   Through  related
techniques, generalized symmetries  of bosonic singletons of arbitrary
spin   have  been   classified   ~\cite{Bekaert:2009fg}  and   concise
formulations of  mixed symmetry higher spin gauge  fields on Minkowski
and AdS spaces have been constructed~\cite{\AGT,\AG}. In this context,
let  us  also  mention  recent  progress  within  the  usual  unfolded
formalism        in        describing       free        mixed-symmetry
fields~\cite{Skvortsov:2008vs,Boulanger:2008up,Boulanger:2008kw,%
  Skvortsov:2009nv}.

For the non linear case treated in the present paper, we take as an
input the antifield dependent BRST differential of a starting point
interacting theory, which encodes the equations of motion, Noether
identities, gauge symmetries and their compatibility
conditions~\cite{Batalin:1981jr,Batalin:1983wj,%
  Batalin:1983jr,Batalin:1984ss,Batalin:1985qj} (see also
\cite{Henneaux:1992ig,Gomis:1995he} for reviews). The parent theory is
then constructed in the form of an extended BRST differential by
introducing the derivatives of the starting point fields as new
independent fields together with associated form fields in such a way
that all additional fields form generalized auxiliary fields.

Various reduced forms can then be obtained from the parent formulation
by eliminating one or another set of generalized auxiliary fields. In
particular, such sets can be related not only to trivial pairs of the
original BRST differential but also to trivial pairs for the extension
of the original BRST differential by the horizontal differential, such
as those studied in \cite{Brandt:1997iu,Brandt:1996mh,Brandt:2001tg}.
The associated generalized tensor calculus and ``Russian formulas''
(see~\cite{Stora:1983ct} for the original derivation) give rise to
corresponding geometrical structures in the reduced formulations.

Besides the obvious connection with unfolding, free differential
algebras and Fedosov quantization, the parent theory can also be
interpreted as a generalized AKSZ sigma model~\cite{Alexandrov:1997kv}
originally proposed in the context of the Batalin-Vilkovisky
formulation for topological field theories (see also
\cite{Cattaneo:1999fm,%
  Grigoriev:1999qz,Batalin:2001fc,Batalin:2001fh,
  Cattaneo:2001ys,Park:2000au,Roytenberg:2002nu,Ikeda:2006wd,
  Bonechi:2009kx,Barnich:2009jy}) for further developments).  More
precisely, the parent differential for non-topological theories
contains an extra term that can however be absorbed by a field
redefinition in the diffeomorphism invariant case.

Although in this paper we restrict ourselves to constructing a parent
formulation for a given gauge field theory for which the interactions
are already known, ultimately our aim is to use the techniques of the
parent formalism to built new interacting models. Indeed, the
usefulness of the parent approach comes from the fact that it combines
in a unified framework the control over the underlying geometry and
the manifest realization of global symmetries of the unfolded
approach~\cite{Vasiliev:2003ev,Vasiliev:2005zu} with the cohomological
control on gauge symmetries provided by BRST theory which leads to
systematic supergeometrical and deformation theoretical
techniques~\cite{Barnich:2000zw,Alexandrov:1997kv,Barnich:1993vg}.

The paper is organized as follows. In Section
\bref{sec:non-lagrangian-aksz}, we quickly review the local field
theory set-up and the BRST differential that describes the gauge
system. We then introduce the necessary additional fields and
operators and provide the parent form of the differential in
Sections~\bref{sec:parent} and \bref{sec:diff-parent-theory}.
Technical details on conventions are relegated to an Appendix. How the
parent theory relates to the AKSZ approach and what it looks like in
the particular case of linear theories is discussed in the next two
sections. General aspects of reductions including cohomological tools
are discussed in Section~\bref{sec:reduction}.  We then illustrate
various features of our analysis on concrete models: we start with
non-degenerate systems, show how the parent formulation for a
relativistic particle reduces to its Hamiltonian formulation, discuss
parent and unfolded formulations of Yang-Mills theory, produce both
off and on-shell versions of parent gravity and finally show how the
parent formulation of the Polyakov string gives rise to a gauge theory
for the Virasoro algebra.

\section{Parent theory}
\subsection{ Original BRST differential}
\label{sec:non-lagrangian-aksz}

The BRST formulation involves bosonic and fermionic fields
$z^\alpha(x)$. The set of fields is graded by an integer degree $g$,
the ghost number $\gh{\ }$. Parity is denoted by $\p{\ }$. The
physical fields are among the ghost number zero fields, while ghosts
and antifields of the minimal sector are typically in positive and
negative ghost numbers respectively.

Besides the local coordinates on the spacetime manifold, denoted by
$x^\mu$ with $\mu=0,\dots,n-1$ and $\gh{x^\mu}=0$, the $z^\alpha$ and
their derivatives are local coordinates on the fiber of an associated
jet-bundle in an algebraic approach (see
e.g.~\cite{Olver:1993,Andersonbook,Anderson1991,Dickey:1991xa,vinogradov:2001}
for reviews). The $z^\alpha$ and their derivatives are denoted by
$z^\alpha_{(\mu)}$ with $(\mu)=\mu_1\dots\mu_k$ a symmetric
multi-index. Local functions are functions that depend on $x^\mu$,
$z^\alpha$ and a finite number of their derivatives. The total
derivative is defined as the vector field
\begin{equation}
\d_\mu=\dover{}{x^\mu}+z^\alpha_\mu\dover{}{z^\alpha}+
z^\alpha_{\mu\rho}\dover{}{z^\alpha_{\rho}}+\dots\equiv
\dover{}{x^\mu}+\dover{^F}{x^\mu},\label{eq:4}
\end{equation}
where $\dover{^F}{x^\mu}$ denotes the action of the total derivative
on the fields $z^\alpha$ and their derivatives. 
For later use, we note that if we collect
the jet coordinates as the coefficients of a Taylor expansion,
\begin{equation}
  \label{eq:15}
  z^\alpha(x)=\sum_{k=0}\frac{1}{k!} z^\alpha_{\mu_1\dots\mu_k}
  x^{\mu_1}\dots x^{\mu_k}\equiv z^\alpha_{(\mu)}x^{(\mu)},
\end{equation}
the part of the total derivative that acts on the jet-coordinates is
also uniquely defined  through the relation 
\begin{equation}
  \label{eq:16}
   \frac{d}{d x^\mu} z^\alpha(x)=\dover{^F}{x^\mu} z^\alpha(x).
\end{equation}

The dynamics and gauge symmetries of the theory are determined by a
nilpotent BRST differential $s$ of ghost number one defined through
\begin{equation}
s z^\alpha=S^\alpha[x,z]\,, \qquad \commut{s}{\d_\mu}=0\,,\label{eq:9}
\end{equation}
where $S^\alpha[x,z]$ are local functions. The second equations
determines the ``prolongation'' of $s$ on the spacetime derivatives of
the fields. A standard field theoretic way to represent the BRST
differential is through functional derivatives or using the condensed
DeWitt notation (see e.g.~\cite{DeWitt:2003pm}), 
\begin{equation}
  s=\int d^nx\, S^\alpha\vddl{}{z^\alpha(x)}=S^a\vddl{}{z^a}\,,\qquad s
  x^\mu=0 \label{eq:10}
\end{equation}
where $S^\alpha$ is taken as a function of $z^\alpha(x)$ and its usual
derivatives, $d^nx=dx^0\wedge\dots \wedge dx^{n-1}$,
$a=(\alpha,x^\mu)$ and the summation convention includes integration
over $x^\mu$. 

The horizontal complex consists of the exterior algebra of $dx^\mu$
with coefficients that are local functions. Elements of this algebra
are denoted by $\omega[x,dx,z]$, with $dx^\mu$ considered as Grassmann
odd, i.e., as anticommuting with all odd fields, $\gh{dx^\mu}=0$. The
horizontal differential is $d_H=dx^\mu\d_\mu$. We assume that
horizontal forms can be decomposed into field/antifield independent
and dependent parts, $\omega[x,dx,z]=\omega[x,dx,0]+ \hat
\omega[x,dx,z]$. The bi-complex involving the latter is denoted by
$\hat\Omega^{*,*}$. A standard result is then the ``algebraic
Poincar\'e lemma'',
\begin{equation}
\begin{split}
H^k(d_H,\hat\Omega)&=0\quad {\rm  for}\quad 0\leq k< n\,,\\
\omega^n&=d_H\eta^{n-1}\iff\vddl{\omega^n}{z^\alpha}
\equiv\dover{\omega^n}{z^\alpha} 
-\d_\mu \dover{\omega^n}{z^\alpha_\mu}+\dots=0.
\end{split}
\end{equation}
The space of local functionals $\hat\cF$ is defined as
$\hat\Omega^{*,n}/d_H \hat \Omega^{*,n-1}$. Important information on
physical properties of the system is contained in $H^g(s,\hat\cF)$, the
local BRST cohomology groups in ghost number $g$ (see
e.g.~\cite{Barnich:2000zw} and references therein).

When considering the total differential of the bi-complex, $\tilde s= s+d_H$
with degree the sum of the ghost number and the form degree,  another
standard result is the isomorphism 
\begin{equation}
H^g(s,\hat\cF)\simeq
H^{g+n}(\tilde s,\hat\Omega)\label{eq:8}\,,
\end{equation}
where the representative of $H^g(s,\hat\cF)$ is obtained by extracting
the component of top form degree $n$ from a representative of
$H^{g+n}(\tilde s,\hat\Omega)$.

\subsection{Extended space of fields and basic operations}
\label{sec:parent}

If $\Psi^A(x)$ denote the fields of the original formulation, the
fields of the parent formulation are given by
$\Psi^A_{(\lambda)[\nu]}(x)$, where $(\lambda)$ denotes a symmetric
multi-index and $[\nu]$ a skew-symmetric one. The fields without
indices are identified with the fields of the original formulation,
$\Psi^A_{()[]}(x)\equiv \Psi^A(x)$. By introducing additional
Grassmann even variables $y^\lambda$, $\gh{y^\lambda}=0$ and Grassmann
odd variables $\theta^\nu$, $\gh{\theta^\nu}=1$, these fields can be
collected in a generalized superfield as follows
\begin{equation}
\begin{split}
\label{decomposition}
\Psi^A(x,y,\theta)&=\sum_{k=0} \sum_{l=0} \,\frac{1}{k!l!}\,\,
\Psi^A_{\lambda_1\dots
  \lambda_k|\nu_1\dots\nu_l}(x)
\theta^{\nu_l}
\dots\theta^{\nu_1}y^{\lambda_k}\dots
y^{\lambda_1} \\
&\equiv\Psi^A_{(\lambda)[\nu]}(x)\theta^{[\nu]}y^{(\lambda)} \,.
\end{split}
\end{equation}
Ghost numbers and parities of $\Psi^A_{(\lambda)[\nu]}(x)$ are then
assigned so that the total ghost number and parity of
$\Psi^A(x,y,\theta)$ is equal to that of $\Psi^A(x)$ by taking the
degrees of $\theta^\nu$ into account, $\gh{\Psi^A_{\lambda_1\dots
    \lambda_k|\nu_1\dots\nu_l}(x)}=\gh{\Psi^A(x)}-l$.  In the
algebraic approach described in the previous section, the $x^\mu$
dependence of the fields is replaced by considering the jet-bundle
coordinates $\Psi^A_{(\lambda)[\nu](\mu)}$ and
\begin{equation}
  \label{eq:17}
  \Psi^A(x,y,\theta)=\Psi^A_{(\mu)(\lambda)[\nu]}
\theta^{[\nu]}y^{(\lambda)}x^{(\mu)}\,. 
\end{equation}
In the parent formulation, the algebra of local functions is taken as
the algebra of functions in $x^\mu,\Psi^A_{(\mu)(\lambda)[\nu]}$,
where each function depends on a finite number of $x^\mu$ derivatives,
i.e., there is no dependence on fields with index $\mu_1\dots\mu_k$
with $k$ strictly greater than some integer.

Consider then the algebra $\cA$ of differential operators acting from
the right in the space of functions in
$x^\mu,y^\lambda,\theta^\nu$. By identifying
$\Psi^A_{(\mu)(\lambda)[\nu]}$ as elements of the basis dual to
$x^{(\mu)}\theta^{[\nu]}y^{(\lambda)}$, one naturally makes linear
functions in $\Psi^A_{(\mu)(\lambda)[\nu]}$ into a left $\cA$
module. Explicitly, if $\cO$  is the linear operator acting on
$x^{(\mu)}\theta^{[\nu]}y^{(\lambda)}$ and 
$\cO^F$ the associated linear operator acting on
$\Psi^A_{(\mu)(\lambda)[\nu]}$, we have
\begin{equation}
  \label{lr-def}
\begin{split}
  \cO^F \Psi^A(x,y,\theta)   = (-1)^{\p{\cO}\p{A}}\Psi^A(x,y,\theta)\cO\,,\\
  \cO^F_1\cO^F_2\Psi^A(x,y,\theta)=(-1)^{(\p{\cO_1}+\p{\cO_2})\p{A}}
\Psi^A(x,y,\theta)\cO_1\cO_2\,.
\end{split}
\end{equation} 
where $\p{\cO}$ and $\p{A}$ is the Grassmann parity of $\cO$ and
$\Psi^A(x,y,\theta)$ respectively.  One then extends this $\cA$-action
to generic functions in $\Psi^A_{(\mu)(\lambda)[\nu]}$ through the
graded Leibnitz rule as a vector field acting from the left, which we
continue to denote by $\cO^F$.

Because both the functions in $x,y,\theta$ and those in
$\Psi^A_{(\mu)(\lambda)[\nu]}$ are modules for $\cA$ considered as a
Lie algebra, the map $\cO \mapsto \cO^F$ respects the graded
commutator,
\begin{equation}
[\cO^F_1,\cO^F_2]=[\cO_1,\cO_2]^F\label{eq:13}\,.
\end{equation} 
Here $\Psi^A(x,y,\theta)[\cO_1,\cO_2]=(\Psi^A\cO_1)\cO_2-
(-1)^{\p{\cO_1}\p{\cO_2}}(\Psi^A\cO_2)\cO_1$.  Some details on the origin
of these conventions are given in the Appendix.

\subsection{BRST differential of parent theory}
\label{sec:diff-parent-theory}

If $s\Psi^A=S^A[x,\Psi]$ with $[s,\d_\mu]=0$ defines the BRST
differential of the original theory, which involves only the
coordinates referring to $x^\mu$ derivatives, the action of $\bar s$ on
$\Psi^A$, $\bar s \Psi^A=\bar S^A[x,\Psi]$ is defined by replacing in
$S^A[x,\Psi]$ the indices corresponding to $x^\mu$ derivatives with the
same indices corresponding to $y^\lambda$ derivatives. It thus follows that
$\bar S^A$ depends on $x^\mu$ and $\Psi^A$'s with indices
$(\lambda)$ corresponding to $y^\lambda$ variables, but no $(\mu)$ nor
$[\nu]$ indices. The action of $\bar s$ is then extended to
$\Psi^A_{(\mu)(\lambda)[\nu]}$ by requiring that
\begin{equation}
  \label{eq:14}
  [\bar s,\d_\mu]=0,\qquad [\bar
  s,\dover{^F}{\theta^\nu}]=0, \qquad [\bar
  s,\dover{}{x^\lambda}+\dover{^F}{y^\lambda}]=0\,. 
\end{equation}
Let $d=\theta^\mu\dl{x^\mu}$ and
$\sigma=\theta^\lambda\dl{y^\lambda}$.  If their action from the right
is defined by $\Psi^A d=\Psi^A\dr{x^\mu}\theta^\mu$ and $\Psi^A
\sigma=\Psi^A \dr{y^\mu}\theta^\mu$, where the arrow denotes right
derivatives, \eqref{lr-def} implies:
\begin{equation}
 d^F\Psi^A=d\Psi^A\,, \qquad \sigma^F\Psi^A=\sigma\Psi^A\,.
\end{equation} 
In particular, when acting on
$\Psi^A_{(\mu)(\lambda)[\nu]}$, $d^F$ and $\sigma^F$ remove an
index from the collection of $[\nu]$ indices and add it to the $(\mu)$
respectively the $(\lambda)$ indices,
\begin{equation}
\begin{gathered}
d^F\Psi^A_{\mu_1\dots\mu_k|(\lambda)|\nu_1\dots\nu_p}=
(-1)^Ap\Psi^A_{\mu_1\dots\mu_k[\nu_1|(\lambda)|\nu_2\dots\nu_p]}\,,\\
\sigma^F
\Psi^A_{(\mu)|\lambda_1\dots\lambda_k|\nu_1\dots\nu_p}=
(-1)^Ap\Psi^A_{(\mu)|\lambda_1\dots\lambda_k[\nu_1|\nu_2\dots\nu_p]}\label{eq:7}\,.
\end{gathered}
\end{equation}

The BRST differential of
the parent theory is the ghost number $1$ operator defined through
\begin{equation}
  \label{parent-BRST}
  s^P=d^F-\sigma^F+\bar s. 
\end{equation}
In order to show that $s$ is nilpotent, one only needs to show that 
\begin{equation}
 \commut{d^F-\sigma^F}{\bar s}=0\,,
\end{equation} 
since both $d^F-\sigma^F$ and $\bar s$ are nilpotent by construction.
Let us first note that by definition of $d^F$ and $\sigma^F$, one
obviously gets
\begin{equation}
 \commut{d^F-\sigma^F}{\bar s}\Psi^A_{(\mu)(\lambda)[\,]}=0\,,
\end{equation} 
if there are no antisymmetric $\nu$ indices. 
Furthermore, 
\begin{equation}
 \commut{\dover{^F}{\theta^\mu}}{\commut{d^F-\sigma^F}{\bar s}}
=\commut{\dlf{x^\mu}-\dlf{y^\mu}}{\bar s}=0\,,
\end{equation} 
where one has used \eqref{eq:13}, the second of \eqref{eq:14} and
finally the difference between the first and the third relation
of \eqref{eq:14}. It then follows that $\commut{d^F-\sigma^F}{\bar s}$
vanishes on all $\Psi^A_{(\mu)(\lambda)[\nu]}$.

To connect to other formulations, let us recall how physical fields,
equations of motion and gauge symmetries are encoded in the parent
formulation (see \cite{\BGST,Barnich:2005ru} for more details).  The
physical fields are among the ghost number $0$ fields where one in
general also finds auxiliary fields and pure gauge degrees of freedom.
In particular, if $\gh{A}$ denotes the ghost degree of $\Psi^A$, then
there are no physical fields associated to $\Psi^A$ if $\gh{A}<0$, and
for $\gh{A_l}=l\geq 0$, there are some among the ghost number $0$
fields $\Psi^{A_l}_{(\lambda)\nu_1\ldots\nu_l}$.  Denoting all the
parent formulation fields at ghost degree $l$ by $\Psi^{\alpha_l}$ the
equations of motion and gauge transformations for the ghost number $0$
fields can be written as
\begin{equation}
  (s^P\Psi^{\alpha_{-1}})\big|_{\Psi^{\alpha_k}=0\,\text{for}\,  k\neq 0}=0\,,
  \qquad
\delta \Psi^{\alpha_0}=(s^P \Psi^{\alpha_0})\big|_{\Psi^{\alpha_k}=0\,\text{for}\, k \neq 0,1}\,,
\end{equation} 
where for the former, all fields with ghost number different from zero
are put to zero, while for the latter, one keeps in addition to the
ghost number $0$ fields, those in ghost number $1$ which are replaced
by gauge parameters. In a similar way, reducibility relations for
equations of motion and for gauge transformations can be read off from
$s^P$ in the sector of fields of higher positive and negative ghost
numbers.

\subsection{AKSZ-type sigma model}
\label{sec:aksz-type-sigma}

The structure of the parent theory BRST
differential~\eqref{parent-BRST} is very similar to that defining AKSZ
sigma models~\cite{Alexandrov:1997kv}.  More precisely, the BRST
differential of the non-Lagrangian version of an AKSZ sigma model is
defined by
\begin{equation}
 s^{AKSZ}\Psi^a(x,\theta)=d \Psi^a(x,\theta)+Q^a(\Psi(x,\theta))\,,
\end{equation} 
where $Q^a$ are the components of an odd nilpotent vector field on the
space with coordinates $\Psi^a$.

It is then straightforward to see that an AKSZ sigma model
corresponds to a parent theory for which all fields with $\lambda$
indices related to $y^\lambda$ derivatives vanish and $\sigma^F$ is
absent. Furthermore, $\bar s$ is required not to involve the
space-time coordinates explicitly. In other words, $\bar s \Psi^a$ is
restricted to be a function of $\Psi^b$ alone.

Whereas one can freely change coordinates $x^\mu$ of the base space of
an AKSZ sigma model without affecting the differential provided the
$\theta$ prolongations of the fields $\Psi^a$ transform tensorially,
this is no longer true for a generic parent differential due to the
$\sigma^F$ term. Note however that a parent differential, for which
this term is absent and there is no explicit $x^\mu$ dependence, is of
AKSZ-type if one takes the indices $a$ to be given by the collection
$A(\lambda)$. As we will show below, this is the case for
diffeomorphism invariant theories, after a suitable field
redefinition. It thus follows that changes of coordinates in the base space
together with the associated tensorial transformation laws for the
fields, do not affect this type of parent differentials either.

\subsection{Linear theories and first quantized description}
\label{sec:linear-theories}

Let us illustrate the construction for linear theories and connect to
the formulation given in~\cite{Barnich:2004cr}. As a first step, one
introduces an auxiliary superspace $\cH$ whose basis elements $e_A$
are associated to the fields $\Psi^A$ and defines
$\gh{e_A}=-\gh{\Psi^A}$, $\p{e_A}=-\p{\Psi^A}$.  The space of
$\cH$-valued space-time functions can be then regarded as the space of
states of a BRST first-quantized system. Indeed, using the string
field $\Psi(x)=\Psi^A(x) e_A$, one can define the first quantized BRST
operator acting from the right according to
\begin{equation}
 s \Psi=  \Psi\brst\,,\qquad (\phi^A e_A)\brst
=\phi^A(x)\brst_A^B(\dr{x},x)e_B\,.\label{1stqant}
\end{equation} 
The nilpotency of $s$ then implies the nilpotency of the operator
$\brst$ and vice-versa.  In addition $\brst$ carries a unit ghost
number and hence determines a first quantized BRST system. The
starting point field theory then appears as the gauge field theory
associated to this first-quantized system \cite{Barnich:2003wj}.

In first quantized terms, the parent theory is obtained by extending
$\cH$ to $\cH^T$ through tensoring with the Grassmann algebra
generated by variables $\theta^\mu$ and formal power series in
$y^\mu$. On the $\cH^T$-valued functions, the BRST operator that gives
rise to the parent differential $s^P$ through \eqref{1stqant} is
defined by 
\begin{equation}
 \brst^T=d-\sigma+\bar\brst\,,
\end{equation}  
where $\bar\brst$ denotes the starting point BRST operator $\brst$
extended to act on $\cH^T$ and with $\dl{x^\mu}$ replaced with
$\dl{y^\mu}$ and $x^\mu$ with $x^\mu+y^\mu$.  

\subsection{Reductions}
\label{sec:reduction}

The usefulness of the parent theory has to do with the possibility to
arrive at other equivalent formulations just by eliminating one or
another set of generalized auxiliary fields. A practical way to
identify such fields is obtained by relating them to a part of the
parent differential $s^P$ using standard homological techniques.

\subsubsection{Generalized auxiliary fields and algebraically trivial
  pairs}
\label{sec:gene-auxil-fields}

Let us briefly recall the notion of generalized auxiliary fields at
the level of equations of motion.  Suppose that, after an invertible
change of coordinates possibly involving derivatives, the set of
fields $z^\alpha$ splits into $\varphi^i,w^a,v^a$ such that the
equations $sw^a|_{w^a=0}=0$, understood as algebraic equations in the
space of fields and their derivatives, are equivalent to
$v^a=V^a[\varphi]$ in the sense that they can be algebraically
solved for fields $v^a$. Fields $w,v$ are then called generalized
auxiliary fields. In the Lagrangian framework, fields $w,v$ are in
addition required to be second-class constraints in the antibracket
sense.  In this context, generalized auxiliary fields were originally
proposed in~\cite{Dresse:1990dj}. Generalized auxiliary fields
comprise both standard auxiliary fields and pure gauge degrees of
freedom as well as their associated ghosts and antifields.

As explained in section 3.2 of \cite{Barnich:2004cr}, there is a
reduced differential associated to the surface defined by the
equations 
\begin{equation}
  \label{eq:19}
  w^a=0,\qquad v^a-V^a[\varphi]=0\,.
\end{equation}
This reduced differential is defined on the space of fields
$\varphi^i$ and their derivatives through $s_R
\varphi^i=s\varphi^i|_{w^a=0,\, v^a=V^a[\varphi]}$. 

By following the reasoning in the proof of proposition~{\bf 3.1} of
\cite{Barnich:2004cr}, it can then be shown that there exists an
invertible change of fields from $z^\alpha$ to
$w^a,sw^{a},\varphi^i_R$ such that $s\varphi^i_R=
S_R^i[\varphi_R]$. If in addition this change of variables is local,
the fields $w^{ a},sw^{a}$ are called algebraically trivial
pairs. Conversely, it follows directly from the form of the
differential in these variables that $w^{ a}$ and $sw^{a}$ are
generalized auxiliary fields. Finally, it can easily be shown that
$\varphi^i_R$ and $\varphi^i$ differ only by terms that vanish when
$w^a,sw^a$ vanish and then that $s\varphi^i_R$ and $s_R \varphi^i$
agree when $w^a$ and $sw^a$ vanish. In other words, in this case the
reduced theories are identical and the concepts of algebraically
trivial pairs and generalized auxiliary fields are the same.

If we restrict to local generalized auxiliary fields, the BRST
cohomology, both the standard and the one modulo the horizontal
differential $d_H$, of $s$ involving all the fields $z^\alpha$ is
isomorphic to the one involving the fields $\varphi^i_R$ alone, or,
what is the same, to the one of $ s_R$ involving $\varphi^i$
alone. Since these cohomology classes contain relevant physical
information (see e.g.~\cite{Barnich:2000zw} for details and
references), it is natural to consider as equivalent the original and
the reduced theories.

In many cases the new variables $\varphi_R^i$ are indeed local
functions. This includes for instance all linear systems but it does
not need to be so in general. Indeed, the $\varphi^i_R$ are constructed
as power series in the variables $w^a,v^a$ which do not necessarily
terminate or sum up to local functions. The elimination of
generalized auxiliary fields is then not a strictly local procedure
and can in principle affect the local BRST cohomology groups. Except
for the case of parametrized theories discussed in
Section~\bref{sec:param}, we only consider local generalized auxiliary
fields. The generalized auxiliary fields relating the parametrized and
non-parametrized formulations of the same theory are manifestly
nonlocal. This is so because through parametrization, one reformulates the
theory in terms of constants of motion which are by construction
nonlocal expressions of the original variables (see
e.g.~\cite{Henneaux:1992ig}).

\subsubsection{Degrees}
\label{sec:degrees}

By using appropriate degrees, generalized auxiliary fields can be
identified by focusing on only a part of the BRST differential. 

A technical assumption satisfied in all models of interest is that the
original BRST differential and thus also the parent differential $s^P$
does not contain constant terms, or in other words, that there is no
term of degree $-1$ in an expansion in terms of homogeneity in the
fields.

Let us more generally assume that the space of fields carries a
suitable degree such that the degree of the independent fields is
bounded from above and that the decomposition of the BRST differential
has a lowest degree $s_p$, which we take for definiteness to be
$p=-1$,
\begin{equation}
 s=s_{-1}+s_0+\ldots\qquad \deg{s_k}=k,
\end{equation} 
where $s_{-1}$ commutes with the total derivative $\d_\mu$. Note that
the considerations below remain true for different values of $p$. We
then have:

\begin{prop}
  Algebraically  trivial pairs for $s_{-1}$ are generalized
  auxiliary fields of the theory determined by $s$.
\end{prop}
\proof{
Indeed, by assumption there is a coordinate system $(w^a,v^a,
\varphi^i_R)$, such that $s_{-1}w^a=v^a$ and
$s_{-1}\varphi^i_R= S_R^i[\varphi_R]$. It follows that
\begin{equation}
 sw^a=v^a+\sum_{k=0} s_k w^a.
\end{equation} 
Equations $sw^a|_{w^a=0}=0$ can be uniquely solved with respect to
$v^a$. To see this, suppose that $v^{a_m}$ are the variable(s) having
the maximal degree $m$ among $v^a$-variables. Variables $w^{a_m}$ have
degree $m+1$ while the terms $\sum_{k=0} s_k w^{a_m}$ have degree $m+1$
or higher and hence the linear part of these terms cannot involve
$v$-variables. Repeating the argument for degree $m-1$ and
lower shows that $sw^a|_{w^a=0}=0$ can be solved to linear order.
In the space of formal series in variables $w,v$ the
equation $sw^{a}|_{w^{a}=0}=0$ can be then uniquely solved with
respect to $v^{a}$ order by order.\qed }

The theory determined by $s$ can thus be reduced to the one involving
the fields $ \varphi^i_R$ alone by eliminating generalized
auxiliary fields.  Suppose now that $s_{-1}$ takes the form
$s_{-1}=v^a\dl{w^a}$. If in addition the cohomology of $s_{-1}$ is
concentrated in degree $0$,  the BRST differential of the reduced
theory has the particularly simple form
\begin{equation}
   s_R \varphi^i_R=(s\varphi^i_R)|_{v=0}=(s_0 
\varphi^i_R)|_{v=0}.
  \end{equation} 
 Indeed, in this case, the equation $sw^a|_{w^a=0}=0$ is solved by
$v^a=0$. This is so because its solution, which has the form
$v^a=V^a[\varphi_R]$ for some local functions $V^a$, implies that
those $v^a$ that have nonvanishing degree vanish as the degree of all
the $\varphi^i_R$ is zero. Let now $v^{a_0}$ denote those $v^a$
that have vanishing degree. Equation $sw^{a_0}|_{w=0}=0$ gives $
-v^{a_0}=(s_0 w^{a_0}+s_1 w^{a_0}+\ldots)|_{w^a=0}$. But all the terms
on the right hand side have positive degree and hence cannot be local
functions of $\varphi^i_R$ unless they vanish.

Note that it can be useful not to identify and eliminate the variables
$v^{a_0}$ with vanishing degree explicitly and to keep them in the
reduced theory.  Indeed, $s_0 v^{a_0}$ can only depend on variables
$\varphi^i_R$ and $v^{a_0}$. By eliminating all variables $w,v$
but $v^{a_0}$, the reduced theory is determined by the BRST
differential $s_0$ and the constraints $v^{a_0}=0$.

\subsubsection{Target space reductions}
\label{sec:targ-space-reduct}

For the case of the parent theory, consider the negative of the
target-space ghost number, i.e., the prolongation to the parent theory
of the original ghost number that does not take into account the
number of $\nu$ indices. It follows that $s^P=s^P_{-1}+s^P_{0}$, where
$s^P_{-1}=\bar s$ and $s^P_0=d^F-\sigma^F$. Using this degree which is
bounded from below, it follows in particular that:
\begin{prop}\label{prop:2.2}
  Algebraically trivial pairs for the original BRST differential $s$
  give rise in the parent formulation to a family of generalized
  auxiliary fields comprising all their descendants obtained through
  $\dover{^F}{\theta^\nu}$ and $\dover{^F}{y^\lambda}$
  derivatives.
\end{prop}
In this case, the reduction of the fields
$\Psi^A_{(\mu)(\lambda)[\nu]}$ thus involves only the $A$
indices and the reduced theory is simply the parent extension of the
reduction $s_R$ for the original BRST differential $s$.

Furthermore, consider trivial pairs for the BRST differential $s$ that
are not necessarily algebraic. In other words, the separation of the
jet-coordinates into trivial pairs and the remaining coordinates does
not respect the differential structure encoded in $\d_\mu$.  In terms
of the parent theory, the same trivial pairs for $\bar s$ in terms of
$y^\lambda$ derivatives are now algebraic as they do not involve the
$x^\mu$ derivatives at all. Of course, the separation now does not
respect the differential structure with respect to the $y^\lambda$
derivatives, but still does for the $\theta^\nu$ derivatives. We thus
also have:
\begin{prop}
  Trivial pairs for $s$ that are not necessarily algebraic give rise
  to a family of generalized auxiliary fields of the parent theory
  which comprises all their descendants obtained through
  $\dlf{\theta^\nu}$-derivatives.
\end{prop}

\subsubsection{Going on-shell}
\label{sec:going-shell}

Consider now a gauge theory described by an antifield dependent BRST
differential that is expanded according to the antifield number,
\begin{equation}
 s=\delta+\gamma+s_1+\ldots\,,
\end{equation} 
where $\delta$ denotes the Koszul-Tate differential
\cite{Fisch:1990rp} which satisfies standard regularity assumptions
(see \cite{Henneaux:1991rx} for the case of local field theories).  In
this case, $\delta=v^a\dover{}{w^a}$ but the point is that the trivial
pairs for $\delta$ are not algebraic. The contracting homotopy $\rho =
w^a\dover{}{v^a}$ does not commute with the total derivative $\d_\mu$
\cite{Barnich:1994db}. If one concentrates on the BRST cohomology in
the space of local functions or horizontal forms, this is not an issue
and the cohomology of $\delta$ is indeed concentrated in degree $0$.
Not eliminating $v^{a_0}$ explicitly corresponds to the case of
considering the weak cohomology of $\gamma$, i.e., the cohomology of
$\gamma$ modulo the relations imposed by the equations of
motions. Note however that in the case of local functionals there
generically is non-trivial cohomology of $\delta$ in strictly negative
degree because the trivial pairs are not algebraic, see
e.g.~\cite{Barnich:2000zw,Barnich:1999cy} for the relation between the
cohomology of $s$ and $\gamma$ in this case.

For the parent theory, using the degree given by minus the
target-space ghost number in a first stage, and then the extension of
the antifield number to the parent theory in a second stage, the piece
$\bar \delta$ of $\bar s$ that corresponds to the prolongation of
$\delta$ is in lowest degree. Moreover the degree of fields is bounded
from above.  In terms of a local field theory in $x$-space, the
elimination of the trivial pairs for its cohomology is now algebraic,
but of course does not respect the $\lambda$ indices in the sense that
it does not commute with $\dover{^F}{y^\lambda}$ in the case of non
trivial equations.

Since the cohomology of $\bar\delta$ is concentrated in antifield
number zero, neither the $\mu$ nor the $\nu$ indices are affected, and
the reduced theory has a differential whose part involving space-time
derivatives $\dover{^F}{x^\mu}$ is still $d^F$ alone, the remainder of
the differential is of antifield number $0$ and no fields of antifield
number different from zero remain. The theory is reduced to the
prolongation of the stationary surface and the variables $v^{a_0}$ are
precisely the lhs of the equations of motion expressed in terms of
$y^\lambda$ derivatives. Keeping them in the formulation allows one not
to choose explicitly independent coordinates on the stationary
surface. Moreover, with these variables kept, the reduced differential
becomes simply $d^F-\sigma^F+\bar\gamma$.
 
\subsubsection{Equivalence of parent and original theory}
\label{sec:equiv-parent-orig}

The main statement that justifies the introduction of the parent
theory is: 
\begin{prop}
  The parent theory determined by $s^P$ can be reduced to the starting
  point theory through the elimination of the generalized auxiliary
  fields.
\end{prop}
\proof{
The proof very closely follows the one for linear systems in
\cite{Barnich:2004cr}. Let $\cN_{\d_x}$ be the operator that counts
the number of $x^\mu$ derivatives on the fields in the original theory,
\begin{equation}
\cN_{\d_x}=\sum_{k=0} k \Psi^A_{\mu_1\dots\mu_k}
\dover{}{\Psi^A_{\mu_1\dots\mu_k}}\label{eq:2}.
\end{equation}
By assumption, the original differential $s$ is local in the sense
that $s$ involves a finite total number of derivatives, or in other
words, the decomposition of $s$ into homogeneous components of
$-\cN_{\d_x}$ is bounded from below, say by $-T$. It then follows that
the same is true for $\bar s$ in the parent theory in terms of
$y^\lambda$ derivatives on the fields counted by the operator
\begin{equation}
\cN_{\d_y}=\sum_{k=0}k \Psi^A_{(\mu)\lambda_1\dots\lambda_k[\nu]}
\dover{}{\Psi^A_{(\mu)\lambda_1\dots\lambda_k[\nu]}}\label{eq:3}.
\end{equation}
The grading is then choosen as $-\cN_{\d_y}$+ T$\times$ (target-space
ghost number). It follows that the lowest part of $s^P$ is $-\sigma^F$
in degree $-1$, while $d^F$ is in degree $0$ and $\bar s$ contains
terms that are of degree greater or equal to zero. For the lowest
part, all additional fields of the parent formulation form
algebraically trivial pairs.  Indeed, if $\rho=\dr{\theta^\mu}y^\mu$
and $N_{y,\theta}=\dr{y^\mu}y^\mu+\dr{\theta^\mu}\theta^\mu$, we have
$[\sigma,\rho]=N_{y,\theta}$ which implies the corresponding relation
for the prolongation of these operators acting on the space of fields
due to \eqref{eq:13}. The result then follows by using the standard
homotopy formula.

According to subsection~\bref{sec:degrees}, the additional fields of
the parent theory are thus generalized auxiliary fields. We still have
to show that the reduced differential coincides with the starting
point differential $s$. In order to do so, the fields
$\Psi^A_{(\lambda)[\nu]}(x)$ are split as follows: $\varphi^i(x) =
\Psi^A_{()[]} (x)\equiv \Psi^A(x)$, the fields $w^a(x)$ which form a
basis of the image of $\rho^F$ acting on the space of fields, and the
fields $v^a=\sigma^F w^a$ which by construction form a basis of the
image of $\sigma^F$. The fields $w^a(x),v^a(x)$ can be expressed in
terms of suitable Young tableaux involving the $\lambda$ and $\nu$
indices, but explicit expressions are not needed for the proof.

To compute the reduced differential, we have in a first step to solve
the equations
\begin{equation}
\label{GAF}
 s^P w^a\big|_{w=0}=0
\end{equation}  
with respect to $v^a$.  Consider the degree which counts
the number of skew-symmetric $\nu$ indices,
\begin{equation}
\cN_{\d_\theta}=\sum_{l=0}l \Psi^A_{(\mu)(\lambda)\nu_1\dots\nu_l}
\dover{}{\Psi^A_{(\mu)(\lambda)\nu_1\dots\nu_l}}\label{eq:1}.
\end{equation}
and split the $w^a,v^a$ according to their degree. In particular,
$\sigma^F$ and $d^F$ lower the degree by $1$, $\rho^F$ raises it by
$1$, while $\bar s$ is of degree $0$, as can be seen from
\eqref{eq:14}. It also follows that the $v^a$'s have non-negative degree,
while the $w^a$'s have strictly positive degree. When acting on
$w^a$'s of lowest degree $1$, $\Psi^A_{()|\nu_1}$,
$\Psi^A_{(\lambda_1|\nu_1)}$, $\dots$, \eqref{GAF} gives the sequence
of equations
\begin{equation}
\label{1steq}
\begin{gathered}
\big(  \d_{\lambda_1} \Psi^A-\Psi^A_{\lambda_1|[]}
+(-1)^{\p{A}}\bar s \Psi^A_{()|\lambda_1}\big)\big|_{w=0}=0\,,\\
 \big( \d_{(\lambda_1}
 \Psi^A_{\lambda_2)|[]}-\Psi^A_{\lambda_1\lambda_2|[]}+
(-1)^{\p{A}}\bar s
  \Psi^A_{(\lambda_1|\lambda_2)}\big)\big|_{w=0}=0\,,\\
  \ldots\,.
\end{gathered} 
\end{equation}
Since $\bar s$ is of degree $0$, in each of these equations, the last
term on the left hand side is necessarily proportional to $v^a$'s of
degree $1$ which implies that the above equations can successively be
solved for the $v^a$ of degree $0$ as
$\Psi^A_{\lambda_1\dots\lambda_k|[]}= \d_{\lambda_1\dots\lambda_k}
\Psi^A+O(1)$, where $O(k)$ denotes terms that are proportional to
$v^a$'s of degree $k$.

Let us split the variables $v^a,w^a$ with respect to both degrees
$\cN_{\d_y}$ and $\cN_{\d_\theta}$ such that
$\cN_{\d_y}v^a_{k,l}=kv^a_{k,l}$ and
$\cN_{\d_\theta}v^a_{k,l}=lv^a_{k,l}$ and analogously for $w$. Note that
there are neither $w^a_{k,0}$ nor $v^a_{0,l}$, $v^a_{k,n}$, where $n$ is
the space-time dimension. Working in the space of polynomials in
$v^a_{k,l}$ with $l>0$, let us consider the equations
\begin{equation}
\big ((d^F-\sigma^F+\bar s)w^a_{k,m}\big)|_{w=0}=0.
\end{equation} 
For $m=n$, at linear order in $v^a_{k,l}$, the last term necessarily
involves $v^a_{r,n}$ with $r\geq k$ and hence vanishes. The equation
then expresses $v^a_{k+1,n-1}$ through the derivatives of
$v^a_{k,n-1}$. Because there are no $v^a_{0,l}$ induction in $k$ shows
that $v^a_{k,n-1}$ all vanish at linear order. Repeating the argument
for $m=n-1$ and so on shows that all $v^a_{k,l}$ with $l>0$ vanish at
the linear order. Higher order corrections are necessarily
proportional to $v^a_{k,l}$ with $l>0$ this remains true to all
orders. This shows that $\Psi^A_{\lambda_1\dots\lambda_k|[]}=
\d_{\lambda_1\dots\lambda_k} \Psi^A$.

In the second and last step, we compute the reduced differential, 
\begin{equation}
 s_R
\Psi^A=(s^P\Psi^A)|_{w^a=0,v^a=v^a[\Psi^A]}=
(\bar s\Psi^A)|_{w^a=0,v^a=v^a[\Psi^A]}\,,
\end{equation}
which reduces to the original differential $s$ because $\bar s \Psi^A$
is by definition $s\Psi^A$ where the $x^\mu$ derivatives of $\Psi^A$
are replaced with the corresponding $y^\lambda$ derivatives, but the
latter are precisely the $v^a$'s of degree $0$.\qed}

{\bf Remark 1:} Instead of assuming polynomials in $v^a_{k,l}$, i.e.,~in
some of the fields that carry $\nu$ indices, one can repeat the proof
assuming polynomials in fields with nonvanishing ghost degree. This
assumption can be more natural from the point of view of BRST theory
and would allow for nonpolynomial expressions in form-fields.

{\bf Remark 2:} From the above proof, it appears that it is possible
to eliminate only a part of the contractible pairs $w,v$ for
$\sigma^F$. Namely, one can eliminate only $w_{k,l}$ and $v_{k,l}$
with $k+l \geq M$ for some $M$. Of course in this case, some $v_{k,l}$
with $l>0$ do not vanish anymore but are expressed through derivatives
of the remaining fields. For $M=0$, one recovers the original theory
while for $M$ sufficiently large, lower order equations are unaffected
and remain first order.  By such a consistent truncation, one can
arrive at a first order formulation with a finite number of
fields. This is just the parent theory counterpart of the usual
truncation of the infinite jet space to a finite one in the case of
equations involving a finite number of derivatives.

\subsubsection{Contractible pairs for $\tilde s$}
\label{sec:sigma+bars}

Consider now as a starting point the extended BRST differential
$\tilde s= \theta^\mu\d_\mu+s$ that acts in the space of local
functions with an explicit dependence in $\theta^\mu$.  When
constructing the associated parent differential, one first has to
replace the $x$ derivatives of the fields in $\overline{\tilde
  s}\Psi^A$ by $y$ derivatives. This gives a differential
$\overline{\tilde s}=\theta^\mu(\dl{x^\mu}+\dlf{y^\mu})+\bar s$ acting
on the space of $y$-derivatives $\Psi^A_{()(\lambda)[]}$ of $\Psi^A$.

Treating $x,\theta$ as independent variables and the $y$-derivatives
of $\Psi^A$ as dependent variables, the prolongation of
$\overline{\tilde s}\Psi^A$ to the entire jet space (i.e., to the $x$
and $\theta$ derivatives of $\Psi^A_{()(\lambda)[]}$) is obtained by
using the total $x$ and $\theta$-derivatives, $\partial_\mu,
\d^\theta_\nu=\dl{\theta^\nu}+\dlf{\theta^\nu}$ (see
e.g.~\cite{Olver:1993} for details on
prolongations):~\footnote{Strictly speaking the prolongation should be
  done using $\dl{\theta^\mu}-\dlf{\theta^\mu}$ as a total
  derivative. We use here the prolongation modified by a change of
  signs for the $\theta$-derivatives in order to fit the convention for
  the parent differential used in the rest of the paper.
  Alternatively, consistent signs can be achieved by starting with
  $\tilde s =-d_H+s$ or by exchanging the sign of the $d^F-\sigma^F$ term in
  the parent differential.\label{foot:sign}}
\begin{equation}
{{(\tilde s)}^P}=\theta^\mu (\dover{}{x^\mu}+
\dlf{y^\mu})+d^F-\sigma^F+\bar s\,.\label{eq:xx}
\end{equation}

It follows that the standard parent differential for $s$ is related to ${\bar{\tilde
    s}}$ through
\begin{equation}
\label{eq:27}
s^P={{(\tilde s)}^P}\big|_{\theta=0}\,.
\end{equation}
In other words, the $d^F-\sigma^F$ term of the parent differential is
automatically generated from the parent prolongation of the term
$\theta^\mu\d_\mu$ in $\tilde s$. This property can be used as
follows:

\begin{prop}
  Trivial pairs for $\tilde s$ give rise to a family of generalized
  auxiliary fields comprising all descendants obtained through total
  $\theta$ derivatives at $\theta=0$.
\end{prop}
\begin{proof}
  By assumption, in the original theory there are new independent
  variables $w=w(x,\theta,\Psi)$, $v=v(x,\theta,\Psi)$ such that
  $\tilde s w=v$. In the expression for $w,v$ we replace $x$-derivatives
  by $y$-derivatives so that $(\tilde s)^Pw=v$.
  It then follows from~\eqref{eq:27} that $s^P w|_{\theta=0}=v|_{\theta=0}$.
  Let $w_{\nu_1\ldots\nu_k}=(\d^\theta_{\nu_1}\ldots
  \d^\theta_{\nu_k}w)|_{\theta=0}$ and
  $v_{\nu_1\ldots\nu_k}=(-1)^k(\d^\theta_{\nu_1}\ldots
  \d^\theta_{\nu_k}v)|_{\theta=0}$.
Using $\commut{\d^\theta_\mu}{\tilde s^P}=\dl{x^\mu}+\dlf{x^\mu}$ one finds
\begin{equation}
(d^F-\sigma^F+\bar s)w_\nu=v_\nu+(\dl{x^\nu}+\dlf{x^\nu})w\,,\label{eq:part}
\end{equation}
and similar formulas for higher $w_{\nu_1\ldots\nu_k}$ and
$v_{\nu_1\ldots\nu_k}$.  Using as a degree $\cN_{\d_\theta}$, one
observes that the equations $(d^F-\sigma^F+\bar
s)w_{\nu_1\ldots\nu_k}=0$ can be algebraically solved for
$v_{\nu_1\ldots\nu_k}$, so that $w_{\nu_1\ldots\nu_k}$ and
$v_{\nu_1\ldots\nu_k}$ are indeed generalized auxiliary fields for the
parent theory. Note in particular that, for $x$-independent $w$'s,
equation \eqref{eq:part} implies that the $w_{\nu_1\ldots\nu_k}$ and
$v_{\nu_1\ldots\nu_k}$ are simply contractible pairs for
$-\sigma^F+\bar s$.
\end{proof}

\subsection{Diffeomorphism invariant theories}
\label{sec:diffeo}

Suppose that the starting point theory is diffeomorphism invariant and
that diffeomorphisms are among the generating set of gauge
transformations. By this we mean that there is no explicit $x^\mu$
dependence in the starting point BRST differential, and thus also none
in the parent differential. Furthermore, the starting point theory has
diffeomorphism ghost fields $\xi^\mu$ (replacing the vector fields
parametrizing infinitesimal diffeomeorphisms) among the fields
$\Psi^A$ and the part of $s\Psi^A$ that involves the undifferentiated
$\xi^\mu$ is given by $s^\prime\Psi^A=\xi^\mu\d_\mu\Psi^A$ for all
$\Psi^A$. When suitably prolonged to all derivatives of the fields,
this means that $s=s^\prime +s^{\prime\prime}$ where
$s^{\prime\prime}$ does not depend on the undifferentiated $\xi^\mu$.

At the level of the parent theory, this implies in particular that
$\bar s$ contains $\xi^\lambda\dover{^F}{y^\lambda}$ as the only piece
which depends on the undifferentiated diffeomorphism ghosts
$\xi^\lambda$.  From the prolongation formulas \eqref{eq:14}, it also
follows that, when acting on
$\Psi^A_{(\mu)\lambda_1\dots\lambda_l|\nu_1\dots\nu_k}$, the piece
originating from $s^\prime$ and containing no derivatives of the
diffeomorphism ghosts but one of type $\theta^\nu$ is given by
$(-1)^A k\xi^\lambda_{()()[\nu_1|}\Psi^A_{(\mu)\lambda\lambda_1\dots
  \lambda_l|\nu_2\dots\nu_k]}$.  Note also that this is the only
term in the parent differential that depends on
$\xi^\lambda_{()()\nu}$.

The piece $-\sigma^F$ in the parent differential $s^P$ can then be
absorbed through the field redefinition
\begin{equation}
\xi^\lambda_{()()\nu}\to
\xi^\lambda_{()()\nu}+\delta^\lambda_\nu\label{eq:6}\,.
\end{equation}
Since all other terms of the parent differential are
unaffected, the parent differential in terms of the new fields takes
the form
\begin{equation}
 s^P=d^F+\bar {s}\,,
\end{equation}
where $\bar {s} =\bar { s^\prime} + \bar { s^{\prime\prime}}$ is precisely the
prolongation of the original BRST differential. If one regroups the
$\lambda$ indices corresponding to the $y^\lambda$ derivatives
together with the $A$ indices and considers the $\Psi^A_{(\lambda)}$
as coordinates of a $Q$-manifold, we have:
\begin{prop}\label{prop:2.5}
  The parent formulation of a diffeomorphism invariant theory is of
  AKSZ-type with an infinite-dimensional target space that contains
  all derivatives of the original fields and a $Q$-structure that
  coincides with the starting point BRST differential
\end{prop}

\subsection{Parametrized theories}
\label{sec:param}

As we have seen, the parent formulation is simpler if the starting
point theory is diffeomorphism invariant.  Of course, any theory can
be made diffeomorphism invariant through parametrization. This means
that the independent variables, the coordinates of space-time, become
fields on the same level as the other fields, while new arbitrary
parameters are introduced instead of the original independent
variables.  In this section, we analyze the parent formulation for
parametrized theories.

One way to construct the parametrized parent formulation is to first
make the theory diffeomorphism invariant and then to construct its
parent formulation following the general procedure explained in the
previous sections. Another possibility is to parametrize directly in
the parent formulation by adding extra fields and gauge symmetries.
It turns out to be more economical and instructive to directly build
the parametrized parent formulation from scratch. 

Suppose that the original gauge theory involves a space-time with
coordinates $y^a$, fields $\Psi^A$, and a BRST differential $s$
defined by $s\Psi^A=s^A[\Psi,y]$ and $\commut{\d_a}{s}=0$, where
$\d_a$ denotes total derivative with respect to $y^a$. As before, we
introduce Grassmann odd variables $\xi^a,\,\,\gh{\xi^a}=1$ standing
for $dy^a$ so that horizontal forms become functions of
$\Psi^A_{(a)},y^a,\xi^a$. The space of horizontal forms is equipped
with the total BRST differential $\tilde s=s+\xi^a
\d_a$. Note that we have changed notations with respect to the
considerations in~\bref{sec:non-lagrangian-aksz} because we reserve
$x^\mu,\theta^\mu$ to denote the space-time coordinates and their
differentials after parametrization.

Let us then consider the AKSZ-type sigma model with target space the
extended jet space with coordinates $\Psi^A_{(a)},y^a,\xi^a$ equipped
with the differential $\tilde s$ and source space the extended
space-time manifold with coordinates $x^\mu,\theta^\mu$. We call the
resulting theory the parametrized parent formulation.
\begin{prop}
  The parent formulation as defined in
  Section~\bref{sec:diff-parent-theory} can be obtained through the
  elimination of the following generalized auxiliary fields from the
  parametrized parent formulation:
\begin{equation}
\label{pp-gaf}
 y^a-Y^a(x),\,\, y^a_{\nu_1\ldots \nu_k}\,\,, k>0 \qquad 
\xi^a,\,\, \xi^a_{\nu}+\ddl{ Y^a}{ x^\nu}, \,\, \xi^{a}_{\nu_1\ldots \nu_k}\,\,, k>1\,.
\end{equation} 
Here $Y^a(x)$ define an invertible change of space-time
coordinates. To obtain both formulations in the same coordinates, one
takes $Y^a(x)=\delta^a_\mu x^\mu$.
\end{prop}
\begin{proof}
  It is straightforward to check that fields~\eqref{pp-gaf} can be
  eliminated by imposing the following constraints
\begin{equation}
 \dlf{\theta^{\nu_1}}\ldots \dlf{\theta^{\nu_k}}(y^a-Y^a(x))=0\,, 
\qquad (d^F+\bar{\tilde s})\dlf{\theta^{\nu_1}}\ldots \dlf{\theta^{\nu_k}}(y^a-Y^a(x))=0\,,
\end{equation} 
so that they are indeed generalized auxiliary fields. After the
reduction, the terms in the reduced differential originating from
$\xi^a\dl{y^a}$ in $\tilde s$ give rise to precisely $-\sigma^F$ if
one in addition takes $Y^a=\delta^a_\mu x^\mu$. Finally, the terms
$d^F$ and $\bar s$ remain intact.
\end{proof}
It is important to stress that in contrast to other reductions
considered in this paper, the elimination of variables~\eqref{pp-gaf}
is not a strictly local operation. In addition to the explicit
space-time dependence of the gauge condition $y^a=Y^a(x)$, the
elimination breaks locality in the sense described in
Section~\bref{sec:reduction}. Namely, in the space of local functions,
it is impossible to decouple variables~\eqref{pp-gaf} and the remaining
variables $\Psi^A$ and their $\theta,y$-descendants. Indeed, looking
for a completion $\tilde\Psi^A$ such that $(d^F+\bar{\tilde
  s})\tilde\Psi^A$ is a function of $\tilde\Psi^A$ and their
descendants, one finds that $\tilde\Psi^A$ necessarily involves
derivatives of arbitrarily high order (see~\cite{Brandt:2001tg} for an
algebraically similar example in the context of local BRST cohomology)
and hence such $\tilde\Psi^A$ do not exist in the space of local
functions.

In spite of this nonlocality, the local BRST cohomology of the
parametrized parent formulation is isomorphic to that of the starting
point theory. Indeed, the local BRST cohomology of the AKSZ-type sigma
model with the target space differential being $\tilde s$ is
isomorphic to $\tilde s$-cohomology of the target space local
functions~~\cite{Barnich:2009jy} and hence coincides with that of the
starting point theory.

The parametrized parent formulation can be also used as a shortcut to
parent formulation for diffeomorphism invariant theories.  If the
starting point theory is diffeomorphism invariant $\tilde s$ can be
brought to the form $\xi^a\dl{y^a}+s$ by redefining the diffeomorphism
ghosts by $\xi^a$ (see e.g.~\cite{Barnich:1994db,Brandt:1997iu}). In
this case $\xi^a$ and $y^a$ are algebraically trivial pairs and can be
eliminated so that the parametrized parent formulation reduces to that
of Proposition~\eqref{prop:2.5}.

To complete the discussion of parametrization and to make contact with
the literature, let us show how the parametrized parent formulation
can be seen as a systematic way to obtain a manifestly diffeomorphism
invariant form for theories invariant under some space-time
symmetries.  Without trying to be exhaustive, let us for simplicity
assume that the starting point theory is translation invariant so that
the BRST differential is $y^a$ independent for a suitable choice of
space-time coordinates. One can then consistently drop the
$y^a$-fields in the parametrized parent formulation as these variable
are completely decoupled from the rest. In this case, the reduction to
the usual parent formulation can be seen as imposing the gauge
condition
$\xi^a=0,\xi^a_{\nu}=-\ddl{Y^a}{x^\nu},\xi^a_{\nu_1\ldots\nu_k}=0$ so
that the theory itself and its reduction to the usual description can
be defined without any reference to fields originating from $y$.  The
r\^ole of the $\xi^a$ variables can also be given another
interpretation: for a translation invariant theory the starting point
$\tilde s$ can be considered as acting on the truncated jet space that
does not involve the $y^a$-variables.  Variables $\xi^a$ can then
be interpreted as constant ghosts that take the translation symmetry into
account in the BRST differential.

This has a straightforward generalization to the case where
translations are part of a larger global space-time symmetry algebra
such as the Poincar\'e, AdS or conformal algebras for instance and
results in the formulation where this symmetry algebra is realized
in a manifest way. Formulations of this type are extensively used in
the context of the unfolded
approach~(see.~\cite{Vasiliev:2003ev,Vasiliev:2005zu} and references
therein) and were also used
in~\cite{Barnich:2006pc,Grigoriev:2006tt,Bekaert:2009fg} in the
context of parent-like formulations.

\subsection{Local BRST cohomology}
\label{sec:local-brst-cohom}

It is instructive to see how the BRST and the local BRST cohomology,
which are by construction isomorphic to the ones of the original
theory, appear in the parent formulation.  Let us begin with the
cohomology in the space of local functions. In the case of the parent
formulation it is natural to consider functions that are local in the
sense that they depend on both $x$ and $y$-derivatives of the fields
only up to some finite order. The isomorphism of BRST cohomologies in
the space of local function can be seen as follows: take as a degree
$\cN_{\d_\theta}+\cN_{\d_y}$. The lowest order terms in $s^P$ is
$s^P_{-1}=d^F$. Its cohomology is given by local functions that do not
depend on both $x$ and $\theta$ derivatives of fields. The reduced
differential is simply $\bar s$ restricted to act on the space of
local functions in $x$ and $\Psi^A_{()(\lambda)[]}$. Exchanging the
role of $x$ and $y$ derivatives this complex can be identified with
the starting point BRST complex. 

As briefly explained at the end of
section~\bref{sec:non-lagrangian-aksz}, in order to compute the local
BRST cohomology, one has to compute the cohomology of $\tilde
{s^P}=s^P+d_H$ in the space of horizontal forms. When identifying
$\theta^\nu\equiv dx^\nu$, this simply amounts to including an
explicit $\theta^\nu$ dependence in the space of local functions. To
explicitly verify the isomorphism, let us again take as a degree
$\cN_{\d_\theta}+\cN_{\d_y}$ so that the lowest term in $\tilde{s^P}$
is again $d^F$. Identifying its cohomology with functions in
$\Psi^A_{()(\lambda)[]},x^\mu,\theta^\nu$ and repeating the steps of
the proof of equivalence in section~\bref{sec:equiv-parent-orig} with
the r\^ole of $x^\mu$ and $y^\lambda$ derivatives exchanged, one finds
that the term $\theta^\mu\dlf{x^\mu}$ entering $d_H$ acts in the
cohomology as $\theta^\mu\dlf{y^\mu}$, $\sigma^F$ acts trivially,
while the action of $\theta^\mu\dl{x^\mu}$ is unchanged. Finally the
reduced differential is just $\tilde s$ with the role of $x$ and
$y$-derivatives exchanged.  In order to make sure that this indeed
gives an isomorphism of cohomologies, let us note that a complete
coordinate system can be chosen to contain besides the trivial pairs
for $d^F$ and $\theta^\mu,x^\mu$, the coordinates
\begin{equation}
\tilde\Psi^A_{(\lambda)}= \sum_{l=0}\frac{1}{l!}
\Psi^A_{(\lambda)\nu_1\dots\nu_l}\theta^{\nu_l}\dots\theta^{\nu_1}
\label{eq:12}\,,
\end{equation}
which are local functions satisfying $\tilde{s^P}\tilde\Psi^A=(\bar s
+ \theta^\mu \dlf{y^\mu})\tilde\Psi^A$ and
$\tilde{s^P}x^\mu=\theta^\mu$. In this way one confirms that the
reduced differential is indeed $\tilde s$ with the role of $x$ and $y$
derivatives of the fields exchanged.  In terms of representatives, the
isomorphism sends functions in $x,\theta$, $\d_{(\mu)}\Psi^A$ to the
same functions with $\d_{(\mu)}\Psi^A$ replaced by
$\tilde\Psi^A_{(\mu)}$.

From the above argument, it follows that if one replaces
$\theta^\mu\dlf{x^\mu}$ with $\theta^\mu\dlf{y^\mu}$ in the expression
for $\tilde{s^P}$, the reduced differential obviously remains
intact. Moreover, after this replacement, the extended differential
coincides with $(\tilde s)^P$ from \eqref{eq:xx} and can be seen as
the prolongation of $\tilde s$ with the role of $x$ and $y$
derivatives exchanged, up to the sign conventions discussed in
footnote~\ref{foot:sign}.

It appears more natural to consider such a modified differential as
the extended BRST differential associated to the parent theory because
then the only term that involves $x$-derivatives of the fields is
still $d^F$.

In the context of the extended parent theory, there are now bona fide
$\theta$ dependent combinations of variables and field
redefinitions. For instance, when taking as a degree $-\cN_{\d_\theta}$
which is bounded from above, the term $\tilde s$ (with the role of
$x,y$ derivatives exchanged) is in lowest degree. It follows that
\begin{prop}\label{d+s}
Algebraically trivial pairs for $\tilde s$ give rise to a family of generalized
auxiliary fields for the extended parent theory involving all
descendants obtained trough $\d^\theta_\nu$ and $\dover{^F}{y^\lambda}$
derivatives. Trivial pairs for $\tilde s$ that are not necessarily
algebraic give rise to a family of generalized auxiliary fields
comprising all descendants obtained through $\d^\theta_\nu$
derivatives. 
\end{prop}

For instance, for diffeomorphism invariant theories as discussed in
subsection~\bref{sec:diffeo}, but now considered in the context
of the extended parent formulation, it is most useful to consider the
$\theta$ dependent change of variables
\begin{equation}
  \label{eq:5}
  \xi^\lambda\to \xi^\lambda-\theta^\lambda,
\end{equation}
from the very beginning. Indeed, on the level of $\tilde s$, it allows
one to absorb the field dependent part of $d_H$ into the starting
point BRST differential. It follows that no $\sigma^F$ appears in the
prolongation. This is consistent with the fact that, on the level of
the standard parent theory, the prolongation of \eqref{eq:5} gives
rise to the redefinition \eqref{eq:6} needed to absorb
$\sigma^F$. In terms of the new variables, the extended parent
differential simply becomes $(\tilde s)^P=d^F+\bar
s+\theta^\mu\dover{}{x^\mu}$. The only term that involves
$x^\mu,\theta^\mu$ is the last one. As a consequence, these variables
are trivial pairs that can be eliminated. The extended parent theory
is then simply described by
\begin{equation}
    \label{eq:26}
   (\tilde s)^P_R=d^F+\bar s\,,
\end{equation}
acting in the space of $x^\mu,\theta^\mu$ independent local
functions. 

\section{Examples}
\label{sec:examples-1}

\subsection{Theory without gauge freedom}

Suppose we have a theory without gauge freedom. Let $\phi^k$ denote
the fields of the theory.  In the BV description, there are in
addition antifields $\phi^*_a$ and the BRST differential is determined
by
\begin{equation}
 s\phi^k=0\,,\qquad s\phi^*_a=L_a\,, \qquad 
\commut{s}{\d_\mu}=0\,,
\end{equation} 
where $L_a[x,\phi^k_{(\mu)}]=0$ and their prolongations $\d_{(\mu)}
L_a=0$ are the original dynamical equations determining the so-called
stationary surface in the space of fields and their $x$-derivatives.

In the parent theory, the only fields of ghost number zero are the
$y$-derivatives of the original fields $\phi^k_{(\lambda)}$ as all
antifields carry negative ghost number.  The parent theory equations
of motion are
\begin{gather}
  (\d_\mu-{\dlf{y^\mu}})\phi^k_{(\lambda)}=0\,, \label{first}\\
  {\dlf{y^{\lambda_1}}}\ldots {\dlf{y^{\lambda_l}}} \bar
  L_a=0\,,\label{second}
\end{gather}
where the equations in the second line determine the equivalent of the
stationary surface in the space of $x^\mu$, the fields and their 
$y$-derivatives. Note that ${\dlf{y^\lambda}}$ is a vector field on this space
which does not affect $x^\mu$. Equations~\eqref{second} are obviously
preserved under the action of ${\dlf{y^\lambda}}$ so that the vector field
${\dlf{y^\lambda}}$ restricts to this stationary surface. We use
$\sigma_\lambda$ to denote this restriction.

Let $x^\mu,Q^\alpha,v^i$ denote a new coordinate system replacing
$x^\mu,\phi^k_{(\lambda)}$ such that $Q^\alpha$ can be used as
coordinates on the stationary surface, while $v^i$ are complementary
coordinates that replace the left hand side of the equations
in~\eqref{second}.  In the $Q^\alpha$ coordinate system one has
\begin{equation}
  \sigma_\lambda = \sigma_\lambda^\alpha(Q)\dl{Q^\alpha}\,,\qquad 
\sigma_\lambda^\alpha(Q)=\left[{\dlf{y^\lambda}} Q^\alpha\right]\Big|_{v^i=0}\,.
\end{equation} 
In terms of the new coordinates, Equations~\eqref{second} simply put
$v^i$ to zero, while Equation~\eqref{first} take the form of a
covariant constancy condition 
\begin{equation}
\label{nlunf}
 \d_\mu Q^\alpha-\sigma^\alpha_\mu(Q)=0\,.
\end{equation}

As a simple illustration, let us consider a scalar field on Minkowski
space with a cubic interaction. We refer to~\cite{Shaynkman:2000ts}
for a detailed discussion of the unfolded formulation for a scalar field
(see also~\cite{Vasiliev:2005zu} for an off-shell description). 

If $\Box^y \phi=\eta^{\lambda_1\lambda_2}\phi_{\lambda_1\lambda_2}$,
constraints~\eqref{second} are given by 
\begin{equation}
  \Box^y\phi+g \phi^2=0\,,
\end{equation}
and its prolongations through $y$-derivatives.  As coordinates
$Q^\alpha$ one can take the traceless parts $\phi^T_{(\lambda)}$ of
$\phi_{(\lambda)}$ while the ${\dlf{y^{(\lambda)}}}(\Box \phi+g
\phi^2)$ are the coordinates $v^i$.  In order to write down explicitly
how $\sigma_\lambda$ acts on some of the $Q^\alpha$, we have to use for
instance
\begin{equation}
  \phi_{\lambda_1\lambda_2}=\phi^T_{\lambda_1\lambda_2}+
  \frac{1}{n}\eta_{\lambda_1\lambda_2}\Box^y\phi=
  \phi^T_{\lambda_1\lambda_2}+\frac{1}{n}\eta_{\lambda_1\lambda_2}
\big[(\Box^y\phi+g \phi^2)-g\phi^2\big]\,,
\end{equation} 
where $n$ is a space time dimension. One then finds
\begin{equation}
  \sigma_\lambda \phi^T=(\dlf{y^\lambda}\phi)|_{v^i=0}=\phi_\lambda^T\,, \qquad
  \sigma_{\lambda_1} \phi^T_{\lambda_2}=\phi^T_{\lambda_1\lambda_2}-
\frac{g}{d}\eta_{\lambda_1\lambda_2}\phi^2\,,
  \qquad \ldots\,,
\end{equation}
so that already for $\phi^T_\lambda$, the coefficients of
$\sigma_\lambda$ become nonlinear.  

\subsection{Geodesic motion of point particle}

Let us illustrate the construction on the example of a point particle
moving along a geodesic in a (pseudo)-Riemannian space-(time).  Since
the model is diffeomorphism invariant in one dimension, its equations
of motion can be brought into an AKSZ form according to our general
discussion in section~\bref{sec:diffeo}. In fact, this holds at the
level of the master action as well. Indeed by going on-shell, we will
show that the target space of the AKSZ parent theory can be reduced to
the extended BFV phase space of the model on which the target space
differential is induced by the BRST charge. Furthermore, it is known
that the BV master action associated to canonical BFV gauge theories
with vanishing Hamiltonians are of AKSZ form~\cite{\GD}.

Using an auxiliary field $\lambda$ playing the r\^ole of an einbein
and the notation $\d=\dover{}{\tau}$, the action for geodesic motion
of a point partice is given by
\begin{equation}
  S[X,\lambda]=\half\int d\tau \big[\lambda^{-1} g_{\mu\nu}(X) \d X^\mu
  \d X^\nu+\lambda m^2\big]=\int d\tau \cL\,.
\end{equation} 
The gauge symmetry corresponding to infinitesimal reparametrizations
of $\tau$ acts as
\begin{equation}
  \delta X^\mu =\d X^\mu \epsilon \,,\qquad 
\delta \lambda = \d(\lambda \epsilon)\,,
\end{equation} 
where $\epsilon$ is the gauge parameter.

Promoting the gauge paramater $\epsilon$ to a Grassmann odd ghost
$\xi$ and introducing the antifields $X^*_\mu,\lambda^*,\xi^*$, the
complete starting point BRST differential is given by
\begin{equation}
\label{s-particle}
\begin{split}
  s X^\mu&=\xi \d X^\mu\,,\quad s \lambda =\d{(\xi
    \lambda)}\,,\quad
  s \xi=\xi\d\xi\,,\\
  s X^*_\mu&=\vddl{\cL}{X^\mu}+\d{(\xi X^*_\mu)}\,,\quad s
  \lambda^*=\vddl{\cL}{\lambda}+\xi \d \lambda^*\,,\\ s
  \xi^*&=\xi \d \xi^*+x^*_\mu \d X^\mu -\lambda \d \lambda^*+2
  \xi^*\d \xi\,.
\end{split}
\end{equation}
Since the model is diffeomorphism invariant and the BRST
transformation of each variable contains the time derivative of this
variable, the considerations of the section ~\bref{sec:diffeo} apply
and the parent theory can be expressed in AKSZ form.

Let us recall the results of \cite{Brandt:1997iu}. They state that the
variables $\{\d^q \lambda,\d^q X^*_\mu,\d^q \xi^*,
\,\,q=0,1,\ldots\}$, where $\partial$ denotes the $\tau$ derivatives
along with their $s$ variations, which can be used to replace the
variables $\{\d^{q+2}X^\mu,\d^{q+1}\xi,\d^{q+1}\lambda^*\}$,
form trivial pairs for the extended BRST differential
$\tilde s$. The remaining coordinates are chosen as
\begin{equation}
\label{change}
 \begin{gathered}
 \tau\,,\theta\,,\qquad   X^\mu\,,\qquad
   P^\mu=\lambda^{-1}\dot X^\mu-(\xi+\theta) g^{\mu\nu} X^*_\nu\,,\\
   \eta=-\lambda(\xi+\theta) \,,\qquad
   \cP=-\lambda^*+\lambda^{-1}(\xi+\theta) \xi^*\,.
\end{gathered}
\end{equation}
Note that no $\tau$ derivatives of the remaining variables
appear. Besides $\tilde s \tau=\theta$, the reduced BRST differential
reads
\begin{equation}
\label{phase}
\tilde s X^\mu= -\eta P^\mu\,, \quad \tilde s P^\mu=\eta
\Gamma_{\rho\nu}^\mu P^\rho P^\nu  \,,\quad 
\tilde s \eta=0\,, \quad \tilde s \cP=\half(P^\mu P_\mu-m^2)\,.
\end{equation}
When using the results of Section~\bref{sec:sigma+bars}, it follows
that the parent differential reduces to $d^F$ plus the prolongation of
the differential defined by \eqref{phase}. Before describing the
latter prolongation more explicitly, let us note that $\tilde s$ is
the BFV Hamiltonian BRST differential of the model. Indeed,
introducing the Poisson bracket by $\pb{X^\mu}{P_\nu}=\delta_\nu^\mu$
and $\pb{\eta}{\cP}=1$,
\begin{equation}
\tilde s =\pb{\Omega}{\cdot}\,, \qquad 
\Omega=\half \eta (P^\mu P_\mu-m^2)\,,
\end{equation} 

For the prolongation, one introduces for each of the remaining
coordinates $z^A\equiv(X^\mu,P_\mu,\eta,\cP)$ a coordinate of opposite
Grassmann parity and ghost number differing by $-1$ according to 
\begin{equation}
\tilde X^\mu=X^\mu+P_*^\mu\theta \,,\quad \tilde P_\mu=P_\mu-
X^*_\mu\theta \,,\quad \tilde\eta=\eta+\cP^*\theta\,,
\quad \tilde \cP=\cP+\eta^*\theta\,.
\end{equation} 
The notations here are chosen such that $z^A$ and $z^*_A$ are
conjugated with respect to the antibracket induced by the above
Poisson bracket~(see~\cite{Alexandrov:1997kv,\GD} for the details on
relation of the target space and the field space bracket structures).

It turns out that the parent differential $s^P=(d^F+\bar{\tilde
  s})|_{\theta=0}$ coincides with the BV differential associated to the
first order master action
\begin{equation}
  \label{eq:21}
  S=\int  d\tau\,\big[ P_\mu \dot X^\mu +\cP \dot \eta
  -\{\Omega,z^Az^*_A\}\big]  =
  \int d\tau d\theta\,\big[d\tilde X^\mu \tilde
  P_\mu -d \tilde \eta \tilde \cP -\Omega(\tilde z)\big]\,.
\end{equation}
The associated classical action can be obtained from $S$ by putting to
zero all the variables with nonvanishing ghost degree and is given by:
\begin{equation}
 S_0= \int  d\tau\,\big[ P_\mu \dot X^\mu 
  -\half \cP^*(P^\mu P_\mu-m^2)\big]\,,
\end{equation} 
where $\cP^*$ is to be identified with the Lagrange multiplier of the
Hamiltonian formalism.

\subsection{Parametrized mechanics}\label{sec:param-ex}

Consider a system of ordinary differential equations
\begin{equation}
\label{1deq}
 \dot \psi^A=V^A(\psi,t)\,.
\end{equation} 
In the parametrized version, one considers new fields $e,t$ and
introduces a new independent variable $\tau$. The equations of motion
and gauge symmetries take the form
\begin{equation}
\label{1daksz}
 \dl{\tau}\psi^A=eV^A(\psi,t)\,, \quad \dl{\tau}t =e\,,\qquad 
\delta_\epsilon \psi^A= \epsilon V^A
\,, \quad \delta_\epsilon t = \epsilon \,,\quad \delta_\epsilon e= 
\dl{\tau}\epsilon \,.
\end{equation} 
In the gauge, $t=\tau$, they indeed coincides with the starting point
system.

Let us now show how~\eqref{1daksz} can be arrived at through the
parametrized parent formulation.  As a byproduct, this also shows
that~\eqref{1daksz} in fact defines an AKSZ-type sigma model in 1
dimension.  The BRST description of the dynamics~\eqref{1deq} is
achieved by introducing a ghost $\xi$ and antifields $\cP^A$.
Variables on the extended jet space are
$\overset{n}{\psi}{}^A,\overset{n}{\cP}{}^A,t,\xi$ where the
superscript $n$ denotes the order of derivatives,
i.e.,~$\overset{1}{\psi}{}^A=\dot\psi^A$. The BRST differential is
determined by
\begin{equation}
  \tilde s t=\xi\,, \qquad \tilde s\xi=0 \,, \qquad \tilde s
  \cP^A=\overset{1}{\psi}{}^A-V^A(\psi,t)+\xi \overset{1}{\cP}{}^A\,,
  \qquad 
  \tilde s \psi^A=\xi \overset{1}{\psi}{}^A\,,
\end{equation} 
and the condition that it commutes with the total time
derivative. According to the general prescription of
Section~\bref{sec:param}, the parametrized parent formulation is a 1d
AKSZ-type sigma model whose extended space-time has coordinates
$\tau,\theta$ while the target space coordinates are
$\overset{n}{\psi}{}^A,\overset{n}{\cP}{}^A,t,\xi$.

It is easy to see that $\overset{n}{\cP}{}^A$ and
$\overset{n+1}{\psi}{}^A-\d_n V^A$ for $n\geq 0$ enter trivial pairs
for $\tilde s$ and can be eliminated. In the reduced theory one stays
with just the coordinates $t,\xi,\psi^A$.  The reduced differential is
given by
\begin{equation}
s^{\red}=d^F+\bar Q\,, \qquad Q=\xi (V^A\dl{\psi^A}+\dl{t})\,.
\end{equation} 
Identifying $\dlf{\theta}\xi$ with the field $e$ of the starting point
formulation, it is straightforward to check that
$s^{\red}$ precisely determines the parametrized
system~\eqref{1daksz}.

The fact that parametrized mechanics can be represented as a 1d AKSZ
sigma model is not surprising and seems to be known.~\footnote{In
  particular, it was independently arrived at by~A.~Sharapov whom we
  wish to thank for a related discussion.} Indeed, as it was already
shown in the case of Hamiltonian/Lagrangian systems that any theory
with vanishing Hamiltonian, and thus in particular a parametrized
system, can be reformulated as a 1d AKSZ sigma model~\cite{\GD} (see
also the discussion in~\cite{Barnich:2005ru}). The above example
provides the non-Lagrangian/non-Hamiltonian version of this result and
can of course easily be generalized to include systems with a gauge
freedom.
 
\subsection{Yang-Mills theory}
\label{sec:yang-mills-theory}

The set of fields for Yang-Mills theory are the components of a Lie
algebra valued 1-form $H_\mu$ and ghost $C$ along with their conjugate
antifields $H^{*\mu}_i$ and $C^*_i$, where $i$ is the Lie algebra
index. The BRST differential is given by $s=\gamma+\delta$, 
\begin{equation}
 \begin{gathered}
   \gamma H_\mu=\d_\mu C+\commut{H_\mu}{C}\,, 
\qquad \gamma C=-\half\commut{C}{C}\,.\\
   \gamma H^{*\mu}_i=f^j_{ik} H^{*\mu}_j C^k\,, 
\qquad \gamma C^*_i = f^k_{ji} C^*_k C^j\,,\\
   \delta H^{*\mu}_i=\vdl{H_\mu}{L[H]}\,,\qquad \delta C^*_i=-\d_\mu
   H^{*\mu}_i + f^k_{ji}H^{*\mu}_k H_\mu^j
\end{gathered}
\end{equation} 
where $L[H]={\rm Tr}\ F_{\mu\nu}F^{\mu\nu}$ is the Lagrangian in terms
of the associated curvatures $F_{\mu\nu}=\d_\mu H_\nu-\d_\nu
H_\mu+[H_\mu,H_\nu]$ and $f^k_{ij}$ are the structure constants. Note
that all of the discussion below that does not involve the precise
form of the original equations of motion and their parent
implementation applies to any regular Lagrangian that is gauge
invariant up to a total derivative.

By reducing to the cohomology of $\delta$, the antifields can be
eliminated from the parent theory as explained at the end of
subsection~\bref{sec:going-shell}. The parent theory is then determined
by $\bar \gamma$ 
\begin{equation}
 s^P=d^F-\sigma^F+\bar \gamma\,, 
\end{equation} 
and the algebraic constraints coming from the equations of motions
\begin{equation}
  \label{eq:20}
  \qquad
 \overline{\vddl{L[H]}{H_\mu}}=0\,,
\end{equation}
with $x$-derivatives replaced by $y$-derivatives, together with all
prolongations of these equations obtained by acting multiple times
with $\dlf{y^{\lambda}}$ and $\dlf{\theta^{\nu}}$.

Let us identify explicitly the field content and the equations of
motion of this reduced parent theory. At ghost number zero we have the
fields $(H_\mu)_{(\lambda)|[\,]}(x)$ and $C_{(\lambda)|\mu}(x)$. It is
useful to keep the $y$ variables and to work in terms of the following
generating functions
\begin{equation}
  A(x,y,\theta)=C_{(\lambda)|\mu}(x)
y^{(\lambda)}\theta^\mu \,,\quad B_\mu(x,y,\theta)= (H_\mu)_{(\lambda)|[\,]}(x)y^{(\lambda)}\,.
\end{equation} 
The equations of motion for these fields are given by $s^P\Psi^A=0$
after having put to zero all fields except for those at ghost number
$0$. One thus has to act with $s^P$ on
$C_{(\lambda)|\mu\nu}y^{(\lambda)}(x)$ and
$(H_\mu)_{(\lambda)|\nu}y^{(\lambda)}$ to find
\begin{equation}
\label{YM}
dA=\sigma A+\half\commut{A}{A}\,,\quad dB_\mu=-\dl{y^\mu}A+
\sigma B_\mu+\commut{A}{B_\mu}\,.
\end{equation}
Along with the above algebraic constraints, these equations are
equivalent to the ones of the original Yang-Mills theory.  In the
abelian case, Equations \eqref{YM}, reduce to the spin $1$ sector of
the equations proposed in \cite{\BGST}. In the non-abelian case,
they were proposed in~\cite{Vasiliev:2005zu}. Let us
also mention a closely related formulation in terms of bi-local
fields~\cite{Ivanov:1976zq,Witten:1978xx,Ivanov:1979hh}.

The above parent formulation can be reduced further by eliminating
from the very start contractible pairs for $\tilde\gamma$ as discussed
in \cite{Brandt:1997iu,Brandt:1996mh,Brandt:2001tg}.  For Yang-Mills
theories, these pairs are given for $k \geq 1$ by the variables
$\d_{(\mu_1}\ldots \d_{\mu_{k-1}} H_{\mu_k)}$ and $\tilde \gamma
\d_{(\mu_1}\ldots \d_{\mu_{k-1}} H_{\mu_k)}$ which substitute for
$\d_{\mu_1}\ldots \d_{\mu_{k}} C$. At the same time, as remaining
variables one uses $\tilde C=C+H$, where $H= H_\mu\theta^\mu$ and the
algebraically independent components of the covariant derivatives of
the curvatures, $D_{\mu_1}\dots D_{\mu_{k-1}}F_{\mu_k\nu}$, which are
given by $D_{(\mu_1}\dots D_{\mu_{k-1}}F_{\mu_k)\nu}$ on account of
the Bianchi identities.  In the approach of
\cite{Brandt:1997iu,Brandt:1996mh,Brandt:2001tg}, the former are known
as generalized connections and the latter as generalized tensor
fields. By direct computation, it follows that
\begin{equation}
  \label{eq:23}
  \tilde\gamma \tilde C=-\half[\tilde C,\tilde C]+F\,,\qquad F=\half
  F_{\mu\nu} \theta^\mu\theta^\nu\,,
\end{equation}
which is the celebrated ``Russian formula'' \cite{Stora:1983ct}.
Furthermore, 
\begin{equation}
  \label{eq:24}
  \tilde \gamma (D_{\mu_1}\dots
D_{\mu_{k-1}}F_{\mu_k\nu})=\theta^{\mu_0}D_{\mu_0}D_{\mu_1}\dots
D_{\mu_{k-1}}F_{\mu_k\nu}+[D_{\mu_1}\dots
D_{\mu_{k-1}}F_{\mu_k\nu},\tilde C]\,,
\end{equation}
so that the reduced differential is simply given by \eqref{eq:23} and
\eqref{eq:24} in terms of the over-complete coordinates
$D_{\mu_1}\dots D_{\mu_{k-1}}F_{\mu_k\nu}$.
The standard equations of motion are imposed by extracting all traces
from the independent covariant derivatives of the curvatures
\cite{Torre:1995kb}~: the independent jet-coordinates parametrizing
solution space are given by $[D_{(\mu_1}\dots
D_{\mu_{k-1}}F_{\mu_k)\nu}]^T$, where the superscript $T$ indicates
the trace-free part. Note that in case one does not want to take out
these traces, one has to keep the Koszul-Tate differential acting on
the antifields and their covariant derivatives.

The field content of the fully reduced parent theory is given by the
$\theta$ prolongation of the Lie algebra valued fields $\tilde C$,
$\tilde C+\tilde C_\nu\theta^\nu +\half \tilde
C_{\mu\nu}\theta^{\nu}\theta^{\mu}$, of which only
$A^i_\nu\equiv-\tilde C^i_\nu$ are in ghost number $0$. Note that
there are no more $y$ derivatives of these variables. In addition
there are the $\theta$ prolongations of the over-complete set of
fields given by the covariant derivatives of the curvatures,
$D^y_{\lambda_1}\dots D^y_{\lambda_{k-1}}F^y_{\lambda_k\nu}+
G_{\lambda_1\dots\lambda_k\nu\rho}\theta^\rho+\dots $, of which only
the $\theta$ independent terms are of ghost number $0$.

Applying $\d^\theta_\nu\d^\theta_\mu$ to both sides of the Russian
formula~\eqref{eq:23} and using Equation~\eqref{eq:part} and its
generalizations, one directly gets $ (-\sigma^F+\bar\gamma)\tilde
C_{\mu\nu}=\commut{A_\mu}{A_\nu}- F^y_{\mu\nu}+\ldots$, where $\ldots$
denote terms involving fields with nonvanishing ghost numbers and the
curvature involves $y$ derivatives of $H$.  Using this in
$(d^F-\sigma^F+\bar s)\tilde C_{\mu\nu}=0$ and putting all the fields
of nonvanishing ghost degree to zero gives the first part of the
equations of motion. When contracting indices with $\theta$'s, they
can be compactly written as
\begin{equation}
\label{YMunf1}
dA+\half\commut{A}{A}=F^y\,,
\end{equation} 
and express the equality of the parent form of the $H$
curvature with the $A$ curvature in terms of $x$ derivatives as
dynamical equations. 

In terms of the over-complete set of fields, the derivation of the
remaining equations of motion is straightforward. Applying
$\d^\theta_\rho$ to both sides of Equation~\eqref{eq:24} and using
Equation~\eqref{eq:part} and its generalizations, one gets 
\[
(-\sigma^F+\bar s) G_{\lambda_1\dots\lambda_k\nu\rho} =-D^y_\rho
D^y_{\lambda_1}\dots D^y_{\lambda_{k-1}}
F_{\lambda_k\nu}-[D^y_{\lambda_1}\dots D^y_{\lambda_{k-1}}
F_{\lambda_k\nu},A_\rho]\,,\] and the corresponding equation of motion
reads
\begin{equation}
\label{YMunf2}
 (\d_\rho+\commut{A_\rho}{\cdot})D^y_{\lambda_1}\dots
D^y_{\lambda_{k-1}} F_{\lambda_k\nu}=
D^y_\rho D^y_{\lambda_1}\dots
D^y_{\lambda_{k-1}} F_{\lambda_k\nu}\,.
\end{equation} 
These equations merely equate covariant derivatives of the tensor
fields with respect to $x^\mu$ using $A_\mu$ with such derivatives
with respect to $y^\mu$ using $H_\mu$. 

After an algebraic projection of the latter equations on the
independent coordinates $D^y_{(\lambda_1}\dots
D^y_{\lambda_{k-1}}F^y_{\lambda_k)\nu}$, the reduced theory is known
as the unfolded form of the theory at the off-shell level and has been
constructed in \cite{Vasiliev:2005zu}. In the unfolded approach, the
ghost number zero fields originating from $C$ are known as the gauge
module, while the $D^y_{(\lambda_1}\dots
D^y_{\lambda_{k-1}}F^y_{\lambda_k)\nu}$ form the so-called Weyl
module. The completely reduced on-shell system can be arrived at by
projecting out the traces of $D^y_{(\lambda_1}\dots
D^y_{\lambda_{k-1}}F^y_{\lambda_k)\nu}$ (see~\cite{Torre:1995kb} for
the explicit structure of the projection). For instance, requiring
$D^y_{(\mu}F^y_{\nu)\rho}$ to be totally traceless obviously imposes
the Yang-Mills equations on $A_\mu$ through~\eqref{YMunf1},
\eqref{YMunf2}, and the Bianchi identity. Note however that for higher
tensors, this projection brings in further nonlinear terms.

 \subsection{Metric gravity}
\label{sec:grav}

The BV description of metric gravity involves as fields the inverse
metric $g^{ab}$ and a ghost field $\xi^a$ that replaces the vector
field parametrizing an infinitesimal diffeomorphism, along with their
antifields $g^*_{ab}$ and $\xi^*_a$.  The BRST differential decomposes
as $s=\delta+\gamma$ where
\begin{equation}
\label{grav-delta}
 \delta g^*_{ab}=\vdl{g^{ab}}L[g]\,,\qquad \delta
 \xi^*_c=g^*_{ab}\d_{c}g^{ab}+
2\d_a(g^{ab}g^*_{bc})\,,
\end{equation} 
and
\begin{equation}
\label{grav-gamma}
\begin{gathered}
  \gamma g^{ab}=L_\xi g^{ab}=\xi^c \d_c g^{ab} -g^{cb}\d_c 
\xi^a -g^{ac}\d_c \xi^b \,,\\
  \gamma \xi^{c}=\half \commut{\xi}{\xi}^{c}=\xi^a\d_a \xi^c\,,\\
  \gamma
  g^*_{ab}=-\d_c(g^*_{ab}\xi^c)-g^*_{ac}\d_{b}\xi^c-g^*_{cb}\d_a\xi^c \,,\\
  \gamma \xi^*_c=\d_a(\xi^*_c\xi^a)+\xi^*_a\d_c\xi^a\,.
\end{gathered}
\end{equation}  
We leave open the precise choice of the diffeomeorphism covariant
equations of motion determined by $L$ and only require standard
regularity conditions together with $\gamma L=\d_aj^a$ for some
$j^a$. In this way, we allow for gravitational theories with higher
curvature and/or gravitational Chern-Simons terms.

For metric gravity, $\gamma X$ contains $\xi^a\d_a X$ for any field
$X$, so that the general discussion of Section~\bref{sec:diffeo}
applies. After the field redefinition, the theory is thus of AKSZ type
with target space coordinates $g_{ab},\xi^a,g^*_{ab},\xi^*_a$ along
with their $y$-derivatives. The parent BRST differential takes the
form $s^P=d^F+\bar { s}$. It also follows from the discussion in
Section~\bref{sec:aksz-type-sigma} that one can use generic
coordinates $x^\mu$ and $\theta^\mu$ in the source space without
affecting the target space. As we are going to see, this gives to the
fields of the parent theory a natural geometrical interpretation in
terms of vielbeins, connections and their higher analogs.

Equivalent reduced formulations are obtained by eliminating various
sets of generalized auxiliary fields. For instance, following
Section~\bref{sec:going-shell}, the elimination of the antifields
$g^*_{ab},\xi^*_a$ gives rise to the differential
\begin{equation}
 s^P=d^F+\bar{\gamma}\,,
\end{equation} 
together with the algebraic constraints 
\begin{equation}
\label{grav-const}
\Big(\pr{\dl{y^{a_1}}}\ldots \pr{\dl{\theta^{b_1}}}
\ldots\big(\overline{\vdl{g^{ab}}L[g]}\big)\Big)=0\,.
\end{equation} 
These constraints can also be understood as constraints in the target
space of the AKSZ sigma model. In this case, the $\theta$-derivatives
are to be dropped and the target space becomes the stationary surface
in the jet-space approach in terms of $y$ derivatives. In other words,
the reduced theory is again an AKSZ-type sigma model with a target
space that is the submanifold defined by the constraints
\eqref{grav-const} (with $\theta$-derivatives dropped) in the
supermanifold with coordinates $g^{ab}_{(c)}$ and $\xi^a_{(c)}$. The
associated odd nilpotent vector field is given by $\gamma$. That
$\gamma$ restricts to the submanifold is a consequence of the
covariance of the equations of motion expressed through
$\commut{\delta}{\gamma}=0$.

In ghost number zero, one finds the 0-form fields $g^{ab}_{(c)}$ and
the 1-form fields $A^a_{\mu(c)}$ coming from the component linear in
$\theta^\mu$ in the expansion of $\xi^a_{(c)}$. In order to
write the equations of motion in terms of generating functions, let us
introduce besides $y^a$ additional formal variables $p_b$ and consider
the algebra of polynomials in $y,p$ equipped with the standard Poisson
bracket $\{p_a,y^b\}=\delta_a^b$. The target space coordinates
$g^{ab}_{c_1\ldots c_l}$ and $ \xi^a_{c_1\ldots c_l}$ can then
be encoded in
\begin{equation}
 G=\half g^{ab}_{(c)}y^{(c)}p_ap_b\,, \qquad \Xi= \xi^{a}_{(c)}y^{(c)}p_a\,,
\end{equation} 
and  the action of $\tilde\gamma$ on these coordinates can be
compactly written as 
\begin{equation}
 \gamma \Xi=\half \{\Xi,\Xi\}\,,\qquad  \gamma
 G=\{\Xi,G\}\,.\label{eq:liegrav}
\end{equation}
Indeed, to lowest order in $y$ these are just
formulas~\eqref{grav-gamma} and then one uses an induction in
homogeneity in $y$. In these terms, the nilpotency of $\gamma$ is a
consequence of the graded Jacobi identity for the even Poisson
bracket. It then follows that $\gamma$ is the Chevalley-Eilenberg
differential associated to the Lie algebra of formal vector fields in
the $y$-variables with coefficients in formal symmetric bi-vectors. In
terms of the parent theory fields, $\Xi+A+\dots$ with
$A=A^a_{(c)\mu}y^{(c)}p_a\theta^\mu$ and $G+\dots$, the equations for
the ghost-number-zero fields determined by the parent differential
$s^P=d^F+\bar{\gamma}$ take the familiar form
\begin{equation}
 dA+\half\{A,A\}=0\,,\qquad
 dG+\{A,G\}=0\,,\label{eq:graveom} 
\end{equation} 
which should be supplemented by the algebraic
constraints~\eqref{grav-const}.  In this reformulation, metric gravity
has turned into a gauge theory for the diffeomorphism group since the
gauge field $A_\mu$ takes values in the Lie algebra of vector fields.

The associated linearized equation for spin $2$ gauge fields were
derived in \cite{\BGST} from the parent theory perspective. Note that
the Poisson bracket can be replaced by the associated $*$-commutator
in \eqref{eq:liegrav} and \eqref{eq:graveom} if one allows for a
$p$-independent component in $G$. This corresponds to coupling to an
extra scalar field. With this replacement, equations
\eqref{eq:graveom} are known in the context of Fedosov quantization
and were shown in~\cite{Vasiliev:2005zu} to describe gravity at the
off-shell level. More generally, in~\cite{Vasiliev:2005zu} it was
shown that by allowing for all powers in $p$, one describe the entire
set of symmetric higher spin fields at the off-shell level.  Let us
also mention that understanding gravity as a gauge theory of the
diffeomorphism group dates back to~\cite{Borisov:1974bn}.  A
closely related formulation of gravity was considered
in~\cite{Pashnev:1997xk}.

Finally, let us eliminate additional generalized auxiliary fields
originating from contractible pairs for $\gamma$. In metric gravity,
these are all the derivatives of the ghosts of degree 2 or higher and
all the symmetrized derivatives of the Christoffel
symbol~\cite{Brandt:1997iu} . After this elimination, one stays with
the following variables: $\xi^a,C^a_b= \xi^a_{b}+\Gamma^a_{bc}\xi^c$
at ghost degree $1$ and $g^{ab},{R^d}_{abc},\ldots$ where dots denote
independent components of covariant derivatives of the Riemann tensor
in terms of $y$ derivatives, $D^y_{c_1}\dots D^y_{c_k}
R^b_{a_1a_2a_3}$. The action of $\gamma$ on the ghost variables is
given by\footnote{When reasoning in terms of $\tilde\gamma$, the
  absorption of the $\sigma^F$ term in the parent differential comes
  from the redefinition $\tilde\xi^a=\xi^a+\theta^a$ and \eqref{eq:yy}
  with $\gamma,\xi^a$ replaced by $\tilde\gamma,\tilde\xi^a$ is the
  gravitational Russian formula.}
\begin{equation}
 \gamma \xi^a=\xi^b C_b^a\,,\qquad 
\gamma C_a^b=C_a^c C_c^b +\half \xi^c \xi^d {R^b}_{acd}\,,\label{eq:yy}
\end{equation} 
while
\begin{multline}
  \label{eq:23a}
  \gamma D^y_{c_1}\dots D^y_{c_k}
  R^b_{a_1a_2a_3}=\xi^{c_0}D^y_{c_0}D^y_{c_1}\dots D^y_{c_k}
  R^b_{a_1a_2a_3} -C^b_dD^y_{c_1}\dots D^y_{c_k}
  R^d_{a_1a_2a_3}+\\+C^d_{c_1}D^y_{d}\dots D^y_{c_k}
  R^b_{a_1a_2a_3}+\dots + C^d_{a_3}D^y_{c_1}\dots D^y_{c_k}
  R^b_{a_1a_2d}\,.
\end{multline}
If one performs the change of variables $g^{ab}=\eta^{ab}+h^{ab}$
where $\eta^{ab}$ is the inverse Minkowski metric and considers formal
power series in $h^{ab}$, one can further eliminate $h^{ab}$ and the
symmetric part of $C^{ab}$, where the index has been raised with
$\eta^{ab}$.  The antisymmetric part $C^{[ab]}$ are the ghosts
associated with the Lorentz algebra. 

The ghost number zero fields of the completely reduced theory are then
given by $e_\mu^a=\dlf{\theta^\mu}C^a$ and $\omega_{\mu
  b}^a=\dlf{\theta^\mu}C^a_b$ and by $D^y_{c_1}\dots D^y_{c_k}
R^b_{a_1a_2a_3}$.  In terms of $e^a=e^a_\mu\theta^\mu$ and
$\omega^{[ab]}=\omega^{[ab]}_\mu\theta^\mu$, the equations of motion
of the completely reduced theory take the form
\begin{equation}
  de^a+ \omega^a_b e^b =0\,,\qquad
  d\omega^a_b+\omega^a_c\omega^c_b=\half e^c e^d R_{cda}^b\,,
\end{equation}
\begin{multline}
d (D^y_{c_1}\dots D^y_{c_k} R^b_{a_1a_2a_3})=e^{c_0}D^y_{c_0}D^y_{c_1}\dots D^y_{c_k}
  R^b_{a_1a_2a_3} -\omega^b_dD^y_{c_1}\dots D^y_{c_k}
  R^d_{a_1a_2a_3}+\\+\omega^d_{c_1}D^y_{d}\dots D^y_{c_k}
  R^b_{a_1a_2a_3}+\dots + \omega^d_{a_3}D^y_{c_1}\dots D^y_{c_k}
  R^b_{a_1a_2d}\,,
\end{multline} 
to be supplemented by the algebraic constraints coming from taking
into account the Bianchi identities to select independent variables
among the $D^y_{c_1}\dots D^y_{c_k} R^b_{a_1a_2a_3}$ and the algebraic
constraints \eqref{grav-const} (without $\theta$ derivatives).  This
completes the systematic derivation starting from the parent
formulation of the completely reduced first order gravitational
equations. When reformulated in terms of the independent fields, the
above equations are known as off-shell unfolded equations and were
proposed in~\cite{Vasiliev:2005zu}. Note that the explicit elimination
of the dependent fields and the implementation of the algebraic
constraints~\eqref{grav-const} brings in further nonlinear terms.

In the same spirit, one can construct the parent formulation for
conformal gravity, for which the relevant tensor calculus has been
constructed in~\cite{Boulanger:2004eh}. The same applies to gravity
with nonvanishing cosmological constants formulated as a gauge theory
of the (A)dS group. In the later case the formulation at the off-shell
level can be inferred from the spin-2 sector of the off-shell theory
proposed in~\cite{\Goff} (see also a somewhat related construction
in~\cite{Bonezzi:2010jr}).

\subsection{2d sigma model}
\label{sec:bosonic-string}

As a final example, let us consider a two dimensional sigma model
invariant with respect to both diffeomorphisms and Weyl
transformations. The field content of the BRST formulation is given by
the two-dimensional metric $g_{\mu\nu}$, scalar fields $\varphi^i$,
diffeomorphism ghosts $\xi^\mu$ and the Weyl ghost $C$.  The BRST
differential in the sector of these variables is given by 
\begin{equation}
\begin{gathered}
\gamma g_{\mu\nu}=\xi^\rho\d_\rho g_{\mu\nu}+
\d_\mu \xi^\rho g_{\rho\nu}+\d_\nu \xi^\rho g_{\mu\rho}+Cg_{\mu\nu}\,,\\
\gamma \varphi^i=\xi^\rho \d_\rho \varphi^i\,, \qquad 
\gamma \xi^\nu =\xi^\rho \d_\rho \xi^\nu\,,\qquad 
\gamma C=\xi^\rho \d_\rho C\,.
\end{gathered}
\end{equation}
The action can for instance be taken as
\begin{equation}
  \label{eq:29}
  S_0=\int d^2 x \big[ \half \sqrt{|g|}g^{\alpha\beta}G_{ij}(\varphi)\d_\alpha
  \varphi^i\d_\beta \varphi^j+\half B_{[ij]}(\varphi)\epsilon^{\alpha\beta}\d_\alpha
  \varphi^i\d_\beta \varphi^j\big]\,.
\end{equation}
As before, implementing the corresponding equations with their Noether
identities is done through the antifields and the Koszul-Tate part of
the BRST differential. We will not discuss this part explicitly below.

Just like in the case of gravity, $\gamma$ contains the
$\xi^\mu\d_\mu$ term so that after the redefinition $\tilde \xi^\mu=
\xi^\mu + \theta^\mu$, the parent theory is determined by
\begin{equation}
s^P=d^F+\bar\gamma\,, 
\end{equation}
to be supplemented by the parent implementation of the original
equations of motion. 

We are now going to work locally both in the base and in the target
space and eliminate the generalized auxiliary fields related to the
contractible pairs identified in~\cite{Brandt:1995gu,Brandt:1997iu} to
construct the reduced off-shell parent theory. One first uses the
Beltrami parametrization of the 2d metric and changes the basis for
ghosts accordingly:
\begin{equation}
\begin{gathered}
h=\frac{g_{11}}{g_{12}+\sqrt{g}}\,,\qquad \bar h=\frac{g_{22}}{g_{12}+
\sqrt{g}}\,, \qquad e=\sqrt{g}\\
\eta=\xi^1+\bar h\xi^2\,, \qquad \bar\eta=\xi^2+h \xi^1\,.
\end{gathered}
\end{equation}
We will use the notation $\d=\d_1,\bar\d=\d_2$ below.
As a next step one observes that the metric components and all their
derivatives along with $C,\d\bar\eta,\bar \d\eta$ and all their
derivatives form contractible pairs and hence can be
eliminated. The remaining variables in the ghost sector are 
\begin{equation}
\label{eq:red-variab}
\eta^p=\frac{1}{(p+1)!}\d^{p+1}\eta\,,\qquad 
\bar\eta^{\bar p}=\frac{1}{(\bar p+1)!}\bar\d^{\bar p+1}\bar\eta\,,
\qquad p,\bar p= -1,0,1,\dots\,,
\end{equation}
while the scalar fields and their derivatives $\d^p\bar \d^{\bar p}
\varphi^i$ are replaced by the tensor fields
\begin{equation}
\qquad T^i_{p,\bar p}=(L_{-1})^p(\bar L_{-1})^{\bar p}\varphi^i\,, 
\qquad p,\bar p= 0,1,\dots\,,\label{eq:tensors}
\end{equation}
which satisfy 
\begin{equation}
\label{T-act}
L_q T^i_{p\bar p}= \frac{p!}{(p-q-1)!} T^i_{p-q,\bar p}\quad
\text{for}\quad 
q < p\,, \quad \qquad L_q T^i_{p\bar p}=0 \quad \text{for} 
\quad q \geq p \,,
\end{equation}
with analogous formulae for $\bar L_{\bar q}  T^i_{p\bar
p}$ and where $L_p,\bar L_{\bar p}$ for $p,\bar p=-1,0,1,\dots$ 
satisfy 
\begin{equation}
\commut{L_p}{L_q}=(p-q)L_{p+q}\,, \qquad \commut{\bar L_{\bar p}}{\bar
  L_{\bar q}}=(\bar p-\bar q){\bar L}_{\bar p+\bar q}\,,\qquad
\commut{L_p}{\bar L_{\bar q}}=0\,.\label{eq:vir}
\end{equation}
Furthermore, 
\begin{equation}
  \label{eq:30}
  L_{-1}=\frac{1}{1-h\bar h}\big(\d-h\bar \d-\sum_{{\bar p}\geq 0}\bar
  H^{\bar p}\bar L_{\bar p}+h\sum_{p\geq 0} H^p L_p\big)\,.
\end{equation}
Here $H^p=\frac{1}{(p+1)!}\d^{p+1} \bar h$ and the corresponding
expressions obtained through formal complex conjugation hold for $\bar
L_{-1}$ and $\bar H^{\bar p}$.  The explicit expressions for the
tensor fields are determined by the requirement that the set of
variables \eqref{eq:red-variab}-\eqref{eq:tensors} is closed under
$\gamma$ (see~\cite{Brandt:1995gu,Brandt:1997iu} for details):
\begin{gather}
\label{vir-gamma}
\gamma \eta^p=\half\sum_{q=-1}^{p+1}(p-2q)\eta^q\eta^{p-q}\,,\qquad
\gamma {\bar \eta}^p=\half\sum_{\bar q=-1}^{\bar p+1}(\bar p-2\bar q)
\bar \eta^{\bar q}\bar \eta^{\bar p-\bar q}\,,\\
\label{T-gamma}
\gamma T^i_{p \bar p}=\sum_{q=-1}^p\eta^q L_q T^i_{p \bar p}+
\sum_{\bar q=-1}^{\bar p}{\bar \eta}^{\bar q} {\bar L}_{\bar q} T^i_{p \bar p}\,.
\end{gather}
The above relations can be compactly written using extra variables
$z,\bar z$, the regular vector fields $l_p=\dr{z}z^{p+1},\bar l_{\bar
  p}=\dr{\bar z}\bar z^{\bar p+1}$, $p,\bar p=-1,0,1,\dots$ satisfying
the same algebra as in \eqref{eq:vir} and the generating functions
\begin{equation}
\Xi=\sum_{p=-1}\eta^p l_p\,, \qquad \bar\Xi=
\sum_{\bar p=-1}{\bar \eta}^{\bar p} \bar l_{\bar p}\,,\qquad
T^i=\sum_{p=0}\sum_{\bar p=0}\frac{1}{p!\bar p!} T^i_{p\bar p}z^p z^{\bar p}\,,
\end{equation}
so that $L_pT^i=T^il_p$.  In these terms, Equations \eqref{vir-gamma}
and \eqref{T-gamma} take the form
\begin{equation}
\gamma \Xi=-\half\commut{\Xi}{\Xi}\,, \qquad \gamma T^i= T^i\Xi+ 
T^i\bar 
\Xi\,.
\end{equation}

In the reduced parent theory, the ghost number zero fields come from
the ghost fields $\eta^p(x),{\bar \eta}^{\bar p}(x)$ which give rise
to 1-form gauge fields, $A^p=A^p_\mu(x)\theta^\mu,\bar A^{\bar p}=
{\bar A}^{\bar p}_\mu(x) \theta^\mu$, with $A^p_\nu\equiv -
\eta^p_\nu, \bar A^p_\nu\equiv - \bar \eta^p_\nu$, and the $T^i_{p\bar
  p}(x)$ fields which are $0$ form fields. In terms of generating
functions, $A(x,z)=\sum_{p=-1} A^pl_p$ and $\bar A(x,\bar
z)=\sum_{\bar p=-1}{\bar A}^{\bar p}\bar l_{\bar p}$ and $T^i(x,z,\bar
z)$, the equations of motion of the reduced theory take the form
\begin{equation}
dA+\half\commut{A}{A}=0\,, \quad d\bar A+
\half\commut{\bar A}{ \bar A}=0\,,  \qquad d { T}^i + {T}^iA
+{ T}^i{\bar A}=0\,,
\end{equation}
with $d=\theta^\mu\dl{x^\mu}$. These equations can be considered as
defining the off-shell system. The on-shell version requires in
addition to impose the analog in terms of $y$ derivatives of the original
equations of motion and their prolongations on the 
$T^i(x,z,\bar z)$ fields.

\section{Conclusions}
\label{sec:conclusions}

In this paper, we have shown how to systematically construct a first
order parent theory associated with a generic interacting gauge field
theory described by an antifield dependent BRST differential. We have
then discussed how to obtain various equivalent formulations through
the elimination of generalized auxiliary fields. Our emphasis here has
been the case where the various equivalent formulations are local
field theories in the sense that all functions depend on the fields
and a finite number of their derivatives. Relaxing the locality
requirement is crucial for other types of questions, such as for
instance the reduction to the light-cone description where physical
degrees of freedom are isolated (see e.g.~\cite{Barnich:2005ga} and
references therein for a discussion in BRST theoretic terms), or the
understanding of the relation of the proper BV master action for BRST
first quantized Hamiltonian system and the $\langle\psi,\hat
Q\psi\rangle$ master action \cite{Barnich:2003wj} used in the context
of string field theories
\cite{Thorn:1987qj,Bochicchio:1987zj,Bochicchio:1987bd,Thorn:1989hm}.

Related to this issue, the variables $y^\lambda$ have been auxiliary
in our construction and merely a bookkeeping device for additional
fields in the theory. At the same time, all fields were considered as
fields on the original space-time with coordinates $x^\mu$. But as
suggested by the superfield notation used in
Section~\bref{sec:parent}, one could also consider these fields as
fields on a doubled space-time with coordinates $x^\mu,y^\lambda$, or
even more generally as fields on the superspace with coordinates
$x^\mu,y^\lambda,\theta^\nu$. In this context, it would be interesting
to try to connect the parent formulation with the recently constructed
double field theory \cite{Hull:2009mi} or the bi-local fields used for
the dual formulation of interacting higher spins on $AdS_4$ in
\cite{Koch:2010cy}.

Possible applications of the proposed formalism involve higher spin
theories at the interacting level. In this context, the most striking
results have been obtained using the unfolded
formalism~\cite{Vasiliev:1990en,Vasiliev:2003ev} (see
also~\cite{Sezgin:2002rt,Sezgin:2001ij,Sagnotti:2005ns,Didenko:2009td,%
  Giombi:2009wh,Sagnotti:2010at} and \cite{Bekaert:2005vh} for a
review). We hope to gain a somewhat better control over the theory and
to make geometrical structures manifest by phrasing it in parent
form. This is supported by a concise formulation of nonlinear higher
spin theory at the off-shell level~\cite{\Goff} (see
also~\cite{Vasiliev:2005zu}) that can be understood as an appropriate
AKSZ-type sigma model.

\section*{Acknowledgements}
\label{sec:acknowledgements}

\addcontentsline{toc}{section}{Acknowledgments}

The authors thank I.~Batalin and R.~Stora for discussions and
N.~Boulanger and P.~Sundell for comments on the manuscript. In
addition, M.G.~is grateful to K.~Alkalaev, I.~Buchbinder, V.~Didenko,
E.~Ivanov, E.~Skvortsov, I.~Tyutin, and M.~Vasiliev for discussions.  G.B.~is
supported in parts by the Fund for Scientific Research-FNRS (Belgium),
by the Belgian Federal Science Policy Office through the
Interuniversity Attraction Pole P6/11, by ``\,Communaut\'e fran\c caise
de Belgique - Actions de Recherche Concert\'ees'' and by IISN-Belgium.
M.G.~is supported by the RFBR grant 08-01-00737, RFBR-CNRS grant
09-01-93105. This work was initiated through a stay of M.G.~financed
by the International Solvay Institutes and partially completed during
a visit of G.B.~at the Lebedev Physics Institute supported through the
Dynasty Foundation.

\appendix

\section{Conventions}

Let $\cH$ be a graded superspace with basis $e_\alpha$.  Consider the
associative superalgbera $\cA$ of linear operators acting on $\cH$
from the right, i.e.
\begin{equation}
 \phi(AB)=(\phi A)B\,, \qquad A,B\in \cA\,.
\end{equation} 
The components are introduced according to
\begin{equation}
e_\alpha A=A_\alpha^\beta e_\beta
\end{equation} 
so that $\cA$ is a matrix superalgbera with the following
multiplication law
\begin{equation}
 (AB)^\beta_\alpha=A_\alpha^\gamma B_\gamma^\beta \,.
\end{equation}
Let $\cH^*$ be the dual space to $\cH$ and $\psi^\alpha$ a dual basis
so that
\begin{equation}
 \inner{e_\beta}{\psi^\alpha}=\delta^\alpha_\beta\,.
\end{equation}  
$\cH^*$ is naturally a left module over $\cA$ with the module
structure defined by
\begin{equation}
  \inner{\phi A}{f}=\inner{\phi}{A^F f}\,, 
\qquad \forall \,\,\phi \in \cH,\,\,\, f \in \cH^*\,.
\end{equation} 

The space $\cH$ can be identified with the space of linear functions
in $\psi^\alpha$ considered as supercommuting variables with parity
and grading defined by that of $e_\alpha$,
$\p{\psi^\alpha}=\p{e_\alpha}$ and
$\gh{\psi^\alpha}=\gh{e_\alpha}$. Under this identification $A^F$ is a
linear vector field on the space with coordinates $\psi^\alpha$. In
components, it reads as
\begin{equation}
 A^F=\psi^\alpha A_\alpha^\beta \dl{\psi^\beta}\,.
\end{equation} 
It can be then extended to the algebra of polynomials in $\psi^\alpha$
through the Leinbnitz rule
\begin{equation}
 A^F(fg)=(A^F f)g+(-1)^{\p{A}\p{f}}f A^Fg\,.
\end{equation} 

The above relations can be compactly written by using a distinguished
element $\Psi$ of $\cH^*\tensor \cH$, the latter being understood as a
$\cA$-bimodule such that $\cA$ acts on the first factor from the left
and on second from the right. This element corresponds to the identity
if one identifies $\cH^* \tensor \cH$ with $\cA$ and is given
in components by 
\begin{equation}
 \Psi=\psi^\alpha \tensor e_\alpha. 
\end{equation} 
In what follows we omit the tensor product sign. This is consistent
with identifying $\Psi$ as the identity element in the space of
functions on $\cH$ with values in $\cH$.

Finally, the relation between left and right actions can be compactly
written as
\begin{equation}
 A^F\Psi=\Psi A\,, \qquad A^F B^F\Psi=\Psi A B\,.
\end{equation} 
To make contact with the main text, let us suppose that instead of
$\cH$ we started with $\hat \cH=\cH \tensor \cV$, for some graded
vector space $\cV$ with basis $e_A$.The distinguished element is 
\begin{equation}
 \Psi=\Psi^{A \alpha} (e_\alpha \tensor e_A)\,,
\end{equation} 
and in terms of components $\Psi^A=\Psi^{A\alpha}e_\alpha$, the 
action of an element of $\cA$ satisfies
\begin{equation}
 B^F\Psi^A=(-1)^{\p{B}\p{\Psi^A}}\Psi^A B\,
\end{equation} 
if one consistently applies the usual sign rule.  If $e_\alpha$ stands
for a basis in polynomials in $y,\theta,x$, this gives the definitions
used in the main text.

\addtolength{\baselineskip}{-3pt}
\addtolength{\parskip}{-3pt}

\def\cprime{$'$}
\providecommand{\href}[2]{#2}\begingroup\raggedright\endgroup

\end{document}